\documentclass[3p,times,onecolumn]{elsarticle}

\usepackage{tikz}
\usepackage{amssymb}   
\usepackage{amsthm}    
\usepackage{amsmath}
\usepackage{xcolor}
\usepackage{enumerate} 
\usepackage{graphicx}
\usepackage[english]{babel}
\usepackage{hyphenat}
\usepackage[shortcuts]{extdash} 
\usepackage{kantlipsum}
\usepackage{breakcites}
\usepackage{appendix}

\counterwithin{figure}{section}

\biboptions{authoryear}
\allowdisplaybreaks

\def\K{\mathcal{K}}

\def\KL{\mathcal{KL}}
\def\R{\mathbb{R}}
\def\N{\mathbb{N}}
\def\Z{\mathbb{Z}}

\def\I{\mathcal{I}}
\def\J{\mathcal{J}_{_{\scriptstyle \neq}}}
\def\bJ{\mathcal{J}_{_{\scriptstyle =}}}

\def\mer{\hfill $\circ$}

\def\qed{$\hfill\blacksquare$}
\def\es{\notag \\}
\def\lp{\left(}
\def\rp{\right)}

\newcommand\jlma[1]{{\textcolor{black}{#1}}} 
\newcommand\blue[1]{{\textcolor{black}{#1}}} 
\newcommand\ajv[1]{{\textcolor{black}{#1}}} 
\newcommand\hh[1]{{\textcolor{black}{#1}}} 

\newtheorem{theorem}{Theorem}[section]
\newtheorem{definition}[theorem]{Definition}
\newtheorem{remark}[theorem]{Remark}
\newtheorem{lemma}[theorem]{Lemma}

\newtheorem{assumption}[theorem]{Assumption}

\newtheorem{corollary}[theorem]{Corollary}

\theoremstyle{remark}
\newtheorem{claim}{Claim}

\makeatletter

\usepackage{accents}
\newcommand\munderbar[1]{ \underaccent{\bar}{#1}}
\makeatother

\journal{Journal of the Franklin Institute}

\begin{document}
\begin{frontmatter}

\title{Strong exponential stability of switched impulsive systems \\with mode-constrained switching}
\author[1]{Alexis J. Vallarella}
\author[2]{José Luis Mancilla-Aguilar}
\author[1]{Hernan Haimovich}

\address[1]{Centro Internacional Franco-Argentino de Ciencias de la Informaci\'on y de Sistemas (CIFASIS),
CONICET-UNR, Ocampo y Esmeralda,
Rosario, Argentina. 
{\texttt{\{vallarella,haimovich\}@cifasis-conicet.gov.ar}}}
\address[2]{Departamento de Matemática, Facultad de Ingeniería, UBA, Avda. Paseo Colón 850, Buenos Aires, Argentina. 
{\texttt{jmancil}@fi.uba.ar}}

\begin{abstract}
Strong stability, defined by bounds that decay not only over time but also with the number of impulses, has been established as a requirement to ensure robustness properties for impulsive systems with respect to inputs or disturbances.
Most existing results, however, only consider weak stability, \ajv{where the bounds only decay with time}. 
In this paper, we provide a method for
calculating the maximum overshoot 
and the decay rate 
for strong 
global uniform exponential stability bounds
for nonlinear switched impulsive systems.
We consider the scenario of mode-constrained switching where not all transitions between subsystems are allowed, and where subsystems may exhibit unstable dynamics in the flow \ajv{and/or} jump maps.
Based on direct and reverse mode-dependent average dwell-time and activation-time constraints,
we derive stability bounds that 
can be improved by considering longer switching sequences for computation. 
We provide an example that shows how the results can be employed to ensure the stability robustness of nonlinear systems that admit a global state weak linearization.
\end{abstract}

\end{frontmatter}

\section{Introduction}
\label{sec:introduction}
Impulsive systems are dynamical systems
where the state evolves continuously almost everywhere, according to a flow map,
except for certain isolated time instants
when the state exhibits discontinuous jumps governed by a static jump equation \citep{lakshmikantham1989theory,feketa2021survey}. 
Switched systems arise when there exist a
family of subsystems or modes, along with a rule that determines which one of the subsystems is active at every instant \citep{liberzon_book03}.
Systems that exhibit both impulsive and switching dynamics arise in
robotics, automotive industry,
aircraft and network control among many others \citep{li2005switched}, spiking interest in the study of their stability.
Several works address stability and robustness of different classes of nonlinear
systems with switched and impulsive dynamics
\citep{guan2005hybrid,gao2015further,liu2019exponential}, 
in particular from the Input-to-State Stability (ISS) perspective
 \citep{cai2009,Liu2011899,gao2016input,mancilla2020uniform,9599445}.
For systems with impulsive behavior, 
two different notions of asymptotic stability exist: \textit{weak stability}, for which the decay is ensured as the time advances, and which is the most popular notion of stability, 
and  \textit{strong stability}, which additionally takes into account the number of impulses
that occur within the elapsed time and is standard in hybrid systems theory \citep{goesan_csm09}. 
It was recently shown that the strong stability property ensures robustness with respect to inputs or disturbances, while its weak counterpart does not \citep{haimovich2020nonrobustness,manhai_lcss20}
\ajv{therefore, this stronger concept of stability deserves investigation.}

Stability in switched systems is usually classified into the categories of arbitrary and constrained switching. 
For continuous- and discrete-time systems under arbitrary switching, 
 exponential stability of all its subsystems
 is not sufficient to ensure stability.
Several sufficient (and necessary) conditions
for 
asymptotic or exponential stability under arbitrary switching
exist, see e.g. \citep{MANCILLAAGUILAR200067,lin2009stability,MOLCHANOV198959,dayawansa}.
In cases where stability cannot be ensured under arbitrary switching, 
imposing constraints on the switching signals may help.
One approach is to limit the number of switchings through 
appropriate bounds, such as those known as
dwell- and average dwell-time constraints, with  
later extensions to the more general Mode-Dependent Average Dwell Time
(MDADT) in \cite{6097035}.
Similar approaches have been applied for achieving exponential stability
in various contexts, including
 switched systems with unstable subsystems \citep{xie2013exponential,zhao2014switching,WEI20193102}, 
linear time-varying systems \citep{ju2022exponential},
 systems with time delays 
\citep{he2020exponential,zhan2021exponential}
and systems with destabilizing impulses 
\citep{zhang2019stability}.
Other works consider
the synthesis of stabilizing switching signals \citep{kundu2016generalized,yang2018stabilization}
or analyze %
stability in the context of hybrid systems involving switching and impulsive behaviour by considering dwell-time conditions 
\citep{wang2013finite,liu2012class,briat2017dwell, li2018input}.

Most of the existing results \jlma{only} provide stability bounds for the  \textit{weak} stability properties, do not combine the analysis 
of switching and impulsive behaviors for nonlinear systems or are difficult to compute numerically. 
\blue{Previous works on switched and time-varying impulsive systems show how to assert ISS properties of certain families of impulsive systems. In \cite{mancilla2020uniform}, it is shown that this can be achieved by reducing the problem to assessing the global asymptotic stability of a family of 
scalar systems. Additionally, \cite{mancilla2021integral} proves that strong ISS properties are ensured for systems that can be written 
as a
linear term plus a (nonlinear) perturbation 
satisfying a bound of affine
form on the state,
by ensuring that the linear term admits a strong global uniform exponential stability bound on the norm of the state. Both results ultimately require the knowledge of the maximum overshoot and the decay rate that characterize strong exponential stability. 
\jlma{Although \cite{mancilla2020uniform} provides some results for obtaining bounds for switched impulsive systems, they are limited to systems where the impulse times coincide with the switching times. Therefore, these results do not allow the frequencies of the impulse times and the switching times to be considered separately, excluding the possibility of taking into account the frequency of the impulse times within each mode in the stability analysis.}
}

In this paper, we address these issues by  
presenting a rather general method for computing the maximum overshoot and the decay rate that characterize strong exponential stability for nonlinear systems exhibiting both switching and impulsive dynamics, \jlma{that takes into account both the frequency of impulse times within each mode and the frequency of switching times.}
The system model employed is given by a set of switching subsystems, each of which is itself an impulsive system, and by a switching signal that selects which impulsive system is active at every time instant. The switching signal has the feature of introducing a possible additional 
impulse
at every switching instant and may be constrained so that not every transition between subsystems is admissible.
By considering impulsive subsystems, jumps that may occur at switching instants, and that not every transition between subsystems is possible, we can encompass a wide range of systems. 
The stability results provided are based on switched bounding functions and require
upper or lower bounds
on i) the number of switches of the switching signal with respect to the total elapsed time and on ii) the number of impulse jumps that occur within each mode in relation to the activation time for that mode. %
Additionally, in the presence of unstable subsystems, conditions on
iii) the activation time of each mode are required to prove stability. 
Since the switching signal may be constrained so that not every transition between subsystems is possible, we show that considering longer switching sequences in the computation may yield tighter stability bounds. %
We illustrate the applicability of our results through numerical examples.
In particular, 
we provide an example where the computation of the maximum overshoot and decay rate is employed to
ensure input-to-state stability of a switched time-varying impulsive nonlinear system that admits a global state weak linearization \citep{mancilla2021integral}.

The paper is organized as follows.
In Section~\ref{sec:preliminaries} we state the basic assumptions and results and we
introduce some of the notation and definitions used throughout the paper. In
Section~\ref{sec:mainresults} we present our main results. 
In Section~\ref{sec:example} we provide some numerical examples of switched impulsive systems with unstable subsystems for which we prove 
strong 
global uniform exponential stability. Conclusions are given in Section~\ref{sec:conclusions}.
\ajv{In the appendix, we provide proofs for some necessary lemmas.}

\section{Preliminaries}
\label{sec:preliminaries}

This section introduces the required notation, defines
the considered type of switched impulsive systems, provides the stability definitions employed and ends with some necessary results on bounds for trajectories.

\subsection{Notation}
$\R$, $\R_{\geq 0}$, $\R_{>0}$, $\N$ 
and $\N_0$ denote the sets of real, nonnegative real, positive real, natural and nonnegative integer numbers, respectively. 
Let $\mathrm{id}$ denote the identity function, so that $\mathrm{id}(x)=x$ for all $x$. 
$\lfloor \cdot \rfloor: \R \rightarrow \Z$ denotes the floor function,
$\lfloor s \rfloor :=  \max \{m \in \Z:m \leq s \}$. For $n,L \in \N$ we define
$mod(n,L) := n - L\lfloor n/L \rfloor$. %
For any sequence $\{r_k\}$ of real numbers, %
we take the conventions: $\prod_{k=\ell}^{\bar \ell} r_k = 1$ and $\sum_{k=\ell}^{\bar \ell} r_k = 0$ if $\ell>\bar \ell$.
For every differentiable function $V:\R^n \rightarrow \R$, its gradient at $x\in \R^n$ is denoted by $\nabla V(x) := \frac{\partial V}{\partial x}(x)$. For any function $h$ of a real argument, $h(t^-) := \lim_{s \to t^-} h(s)$ is the left-sided limit of $h$ at $t$.

\subsection{Switched impulsive systems}
\jlma{Let  
 $\I := \{1, \hdots,  N\}$, $\bJ = \{(i,i) : i \in\I \}$ and $\J \subset \{(i,j): i,j \in \I \land i \neq j \}$.}
 \jlma{ Given a family of {\em flow maps} $f_i:\R^n \rightarrow \R^n$, $i\in \I$, that are locally Lipschitz and satisfy $f_i(0)=0$, a family of {\em jumps maps} $g_{i,j}:\R^n \rightarrow \R^n$, $(i,j) \in \J  \cup \bJ$, that satisfy $g_{i,j}(0) = 0$, a finite or infinite sequence $\sigma=\{\tau_1,\tau_2,\ldots\}\subset \R_{> 0}$ of \textit{event times}, which we suppose along the paper strictly increasing and such that $\tau_k\to \infty$ when it is infinite, and a {\em switching signal} $\nu: \R_{\geq 0} \rightarrow \I$ that is right-continuous, piecewise constant and such that the set of switching times
\begin{equation}
s(\nu):=\{t > 0 : \nu(t^{-})\neq \nu(t)\}, \label{eq:switching_times}
\end{equation}
is contained in the set of event times $\sigma$, we consider the switched impulsive systems defined by
\begin{subequations}
\label{impulsive_new}
\begin{eqnarray}
\dot x(t)=f_{\nu(t)} (x(t)) \phantom{lllll|ll} & &t \notin  \sigma,  \label{3a-} \\
x(t)=g_{\nu(t^-),\nu(t)} (x(t^-))& &t \in  \sigma, \label{3b-}
\end{eqnarray}
\end{subequations}
where $t\geq 0$ and $x(t)\in \R^n$ is the state.}
\jlma{The value $\nu(t)\in\I$ is called the switching mode of the system at time $t$. The sequence $\sigma$ contains the switching times, where the switching signal is discontinuous, as well as possibly other instants belonging to the set
\begin{equation}
  \mu := \sigma \setminus s(\nu),
  \label{eq:single-mode-impulses}
\end{equation}
where impulses can also occur, according to (\ref{3b-}) but with $\nu(t^-)=\nu(t)$. 
We will refer to the times in 
$\mu$ as the {\em nonswitching impulse times} to distinguish them from the switching times. Note that the sequence of event times thus satisfies
\begin{equation}
  \sigma= s(\nu)  \cup \mu.
 \label{eq:event_times} 
 \end{equation}
Since  $\J$ is not necessarily equal to $\I \times \I$, the set $\J$ constrains the admissible switching signals $\nu$ as follows:
\begin{equation}
\label{eq:Jconstraint}
   (\nu(t^-),\nu(t))\in \J, \quad \text{for all }t\in s(\nu).
 \end{equation}  
 Note that $\mu$ contains all the impulse times that occur between switching times. If $t\in \mu$ and mode $i$ is active, that is $\nu(t^{-})=\nu(t)=i$, then the jump map which governs the state jump is $g_{\nu(t^-),\nu(t))}=g_{i,i}$.}
 Figure~\ref{fig:WDG_example2} illustrates these sets for a system with $N=3$. 

The type of switched impulsive system $\Sigma$ considered is thus uniquely determined by the flow maps $f_i$ for $i\in\I$, the jump maps $g_{i,j}$ for $(i,j) \in \J\cup\bJ$, the switching signal $\nu$ and the set of nonswitching event times $\mu$, so that we may write
\begin{equation}
  \Sigma=\Big(\{f_i\}_{i\in\I},\{g_{i,j}\}_{(i,j)\in\J\cup\bJ},\nu,\mu\Big). \label{system1}
\end{equation}
\begin{figure}[ht!]
\begin{center}
\begin{tikzpicture}
  \node[draw, circle] (1) at (0,0) {$f_1,g_{1,1}$};
  \node[draw, circle] (2) at (2,0) {$f_2,g_{2,2}$};
  \node[draw, circle] (3) at (1,-2) {$f_3,g_{3,3}$};
  \draw[->] (1) edge[bend right=-45] node[below] {$g_{1,2}$} (2);
  \draw[->] (2) edge[bend right=-45] node[above] {$g_{2,1}$} (1);
  \draw[->] (2) edge node[right] {$g_{2,3}$} (3);
  \draw[->] (3) edge node[left] {$g_{3,1}$} (1);
\end{tikzpicture}
\end{center}
\vspace{-0.4cm}
\caption{Example of the admissible jumps for a switched impulsive system with $N=3$ subsystems, $\J=\{(1,2),(2,1),(2,3),(3,1)\}$ and $\bJ = \{(1,1),(2,2),(3,3)\}$. Switching from mode $1$ to $3$ and from $3$ to $2$ is not possible.}
\label{fig:WDG_example2}
\end{figure}
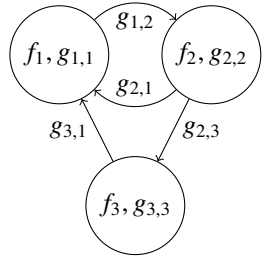
\jlma{In the sequel, let $\Gamma$ denote the set of all possible sequences $\sigma$ of event times.} 

A solution
to \eqref{impulsive_new} is a right-continuous function $x : [t_0 , T_x ) \rightarrow \R^n$ that is locally absolutely continuous on each interval $[a, b)$ such that $(a, b) \cap
\sigma =\emptyset$, its discontinuities are of first kind and satisfies \eqref{3a-} for almost
all $t \in [t_0, T_x) \backslash \sigma$ and \eqref{3b-} for all $t \in (t_0, T_x ) \cap \sigma$.
Note that a solution always starts from $x_0:=x(t_0)$ with a flow, even if $t_0 \in \sigma$.

\begin{remark}
  The formulation~(\ref{impulsive_new}) can model switched and impulsive systems of several different types. 
  For instance, continuous-time switched systems without impulsive behaviour are covered by defining $g_{i,j}=\mathrm{id}$ for all \ajv{$(i,j) \in \J \cup \bJ$} and $\sigma = s(\nu)$. 
  \blue{
Standard discrete-time switched systems defined by $x_{k+1}=h_{\zeta(k)}(x_k)$ 
for all $k\in \N_0$ with $\zeta: \N_0 \rightarrow \I$  
and $h_i:\R^n \rightarrow \R^n$ for all $i\in \mathcal{I}$ can be modelled in the present setting. To see this, define
 $\nu:\R_{\geq 0} \rightarrow \mathcal{I}$ as $\nu(t)=\zeta(k)$ for all $t\in [k,k+1)$.
The switching times are then $s(\nu)=\{t\in \N: \nu(t^-)\neq \nu(t)\}$ and the nonswitching impulse times are
$\mu:= \N \backslash s(\nu)$, so that the sequence of event times is $\sigma=\N$. Define the flow and jump maps as $f_i = 0$ for all $i\in \mathcal{I}$ 
and $g_{i,j}:=h_i$ for all $(i,j)\in \mathcal{J}_= \cup \mathcal{J}_{\neq}$, so that the switched impulsive system in the present scheme becomes 
\begin{align*}
\dot x(t)&=0 \phantom{_{\nu(t^-), \nu(t)}(x(t^-))= h_{\nu(t^-)}(x(t^-))= h_{\zeta(t-1)}(x(t-1))}\quad t\notin \N \\
x(t)&=g_{\nu(t^-), \nu(t)}(x(t^-))= h_{\nu(t^-)}(x(t^-))= h_{\zeta(t-1)}(x(t-1)) \quad t\in \N
\end{align*}
Defining $x_k:=x(t-1)$, where for $t\in \N$ we have $k\in \N_0$, the discrete-time switched system is obtained.
    }
  Additionally, any combination of switched, impulsive, continuous- or discrete-time behaviours can be successfully covered.
\end{remark}

\subsection{Stability definitions}
Given a switching signal $\nu$ and a sequence of nonswitching impulse \blue{times} $\mu$, \blue{the quantities} $n_\nu(t,t_0)$, $n_\mu(t,t_0)$ and $n(t,t_0)$ represent the number of switching, nonswitching and total impulses, respectively,
that lie in the interval $(t_0,t]$, 
\begin{align}
  n_\nu(t,t_0)&:= \# \lp  s(\nu) \cap (t_0,t]  \rp, \\
  n_\mu(t,t_0)&:= \# \lp \mu \cap (t_0,t] \rp,\\
  n(t,t_0) &:=n_\mu(t,t_0)+n_\nu(t,t_0).
\end{align}
\jlma{ \begin{definition}
\label{SWGUAS}
A switched impulsive system $\Sigma$ as in \eqref{system1}
is said to be
Strongly Globally Uniformly Exponentially Stable (S-GUES) if there exist positive constants $K$ and $\lambda$ such that for all 
$x_0 \in \R^n$ and $ t_0 \geq 0$ the
solution with $x(t_0) = x_0$ satisfy 
\begin{equation}
    |x(t)|\leq K|x_0| e^{-\lambda(t-t_0+n(t,t_0))} \quad \forall t\geq t_0. \label{cota2}
\end{equation}
\end{definition}
}
It is often of interest to know whether the same stability bound holds not just for a single system but also for many systems grouped into a family of systems. This family could represent, for instance, all systems sharing the same flow and jump maps but with different switching signals or event-time sequences. With this aim in mind, let $\Lambda$ be a set of %
pairs $(\nu,\mu)$, where $\nu$ is a switching signal and $\mu$ is a sequence of nonswitching impulse times. 
For future reference, we denote the family of switched impulsive systems \eqref{system1} 
for which $(\nu,\mu) \in \Lambda$ as 
\begin{equation}
\Sigma_\Lambda:=\{\Sigma: (\nu,\mu) \in \Lambda\}. \label{system2}
\end{equation}
\jlma{
\begin{definition}
\label{family_GUAS}
The family of switched impulsive systems $\Sigma_\Lambda$ is said to be S-GUES
if there exist positive constants $K$ and $\lambda$ such that for all $(\nu,\mu) \in \Lambda$,
 $x_0 \in \R^n$ and $ t_0 \geq 0$ the
solutions of $\Sigma_\Lambda$ with $x(t_0) = x_0$ satisfy  \eqref{cota2}. 
\end{definition}
}
The above stability properties are all uniform over the initial time $t_0\geq 0$; in addition, for a family of systems $\Sigma_\Lambda$, they are also uniform over all $(\nu,\mu)\in \Lambda$.

\subsection{Bounds on state trajectories}
Along the paper we consider that the following holds.
\jlma{
\begin{assumption}
\label{ass031}
For all $i\in \I$, there exist
continuously differentiable functions $V_{i}: \R^n \rightarrow \R_{\ge 0}$,
 quantities $\munderbar K_i>0$, $\bar K_i>0$, $\bar \lambda(i) \in \R$, 
and quantities $\bar r(i,j)\in \R_{>0}$
for all $(i,j)\in \J \cup \bJ$ 
such that for some $m\in \N$ and all $\xi \in \R^n$ we have
\begin{subequations}
\begin{align}
&\munderbar{K}_{i}|\xi|^m \leq V_{i} (\xi) \leq \bar K_{i}|\xi|^m, \label{ass03-1}\\
&\nabla V_i(\xi) f_i(\xi) \leq \bar \lambda(i) V_{i}(\xi), \label{ass03-2}\\
&V_{j}(g_{i,j}( \xi)) \leq \bar r(i,j) V_{i}(\xi).\label{ass03-3}
\end{align}
\end{subequations}
\end{assumption}
}

Assumption~\ref{ass031} allows to bound the state trajectory of the system of the form  \eqref{impulsive_new} by using different functions $V_i$ for each subsystem $i\in \I$,
without assuming 
\ajv{that the each flow and jump maps
can be bounded
by  $\bar \lambda(i)$ and $\bar r(i,j)$ in \eqref{ass03-2} and \eqref{ass03-3}, respectively, 
such that convergence to the equilibrium 
is guaranteed for each mode.}
As a result, $\bar \lambda(i)$ may be positive and $\bar r(i,j)$ may be greater than $1$ for some $(i,j) \in \J \cup \bJ$. 
In \eqref{ass03-3}, one can distinguish
the values of the form $\bar r(i,i)$, which bound the jumps at impulse times in $\mu$,
and of the form $\bar r(i,j)$ with $(i,j) \in \J$, used for switching times in $s(\nu)$. \ajv{For linear switched systems, functions $V_i$ for each mode that satisfy Assumption~\ref{ass031} can be easily derived; see Appendix~\ref{app:boundingLSI}.}

In what follows, 
we derive sufficient conditions 
that ensure 
strong GUES of a family $\Sigma_{\Lambda}$
under  
Assumption~\ref{ass031}.
This rather general result in Lemma~\ref{cor:1} will form the basis for the results presented.
For this purpose, given a switching signal $\nu$
and the quantities $\bar\lambda(i)$ in \eqref{ass03-2}, define the piecewise constant right\-/continuous function
$\lambda :  \R_{\geq 0} \rightarrow \R$ as
\begin{equation}
 \lambda(t):=\bar \lambda({\nu(t)}).
 \label{func_lambda3} 
\end{equation}
In additional correspondence with \jlma{$\sigma=\{\tau_k\}$} 
and $\bar r(i,j)$ in \eqref{ass03-3}, define the distribution $r$ %
as   
\begin{equation}
r(\cdot):= \sum_{k=1}^\infty \ln \lp \bar r(\nu({\tau_k}^-),\nu({\tau_k}) \rp \delta(\cdot-\tau_k),   \label{func_mu3}
\end{equation}
where $\delta(\cdot-\tau_k)$ is the Dirac delta distribution $\delta$ translated to $\tau_k$. Therefore $\int_{t_0^+}^t r(\tau)d\tau=\sum_{\tau_k\in (t_0,t]} \ln \lp \bar r(\nu({\tau_k}^-),\nu({\tau_k}) \rp$.
Appendix~\ref{appendix:0} contains the proof of Lemma~\ref{cor:1}.

\begin{lemma}
\label{cor:1}
Consider a switched impulsive system $\Sigma$ of the form \eqref{impulsive_new} and suppose that
Assumption~\ref{ass031} holds.
Suppose that there exist $C, \lambda_0>0$ such that
\begin{equation}
\int_{t_0^+}^t\lambda(\tau)+r(\tau) d\tau \leq C- \lambda_0(t-t_0+n(t,t_0)), \quad \forall t\geq t_0 \ge 0\;\forall (\nu,\mu)\in \Lambda.  \label{condi2c}
\end{equation}
Then $\Sigma_{\Lambda}$ is S-GUES with 
\begin{equation}
\label{eq:K_lamb2}
\ajv{
K:= e^{\frac{C}{m}} \sqrt[m]{  \frac{\max_{i\in \I} \bar K_i}{\min_{i\in \I} \munderbar K_i}} \quad \text{and} \quad \lambda:=\frac{\lambda_0}{m}.}
\end{equation}
\end{lemma}

\section{Main results}
\label{sec:mainresults}

In this section we establish sufficient conditions for 
S-GUES for families of switched impulsive systems of the form \eqref{impulsive_new} 
under Assumption \ref{ass031}. 
The conditions are based on mode-dependent average dwell time bounds
for the number of switchings 
and the number of impulsive jumps within each subsystem. 
\ajv{
We first provide a method to improve the system's state bound by appropriately handling the effect of the switchings.}
Subsequently we provide activation time conditions 
that ensure that \jlma{`stabilizing'} modes are active for a sufficiently long time relative to the \jlma{`destabilizing'} modes
and we identify further conditions to guarantee S-GUES, providing
an explicit calculation of
the quantities that define the property. 

\subsection{Average dwell time conditions}

Given a switching signal $\nu$, 
we define the activation time for each subsystem $i\in \I$ in the interval $[t_0,t]$ as 
\begin{equation}
t_a(i,t,t_0):= \int_{t_0}^t \nu_{on} (i,\tau) d\tau, \quad \forall t \geq t_0  \geq 0
\end{equation}
where the activation function %
$\nu_{on}: \I \times \R_{\geq 0} \rightarrow \{0,1\}$
indicates which subsystem is active at a given time
\begin{equation}
\nu_{on}(i,t):= 
\begin{cases}
1 \phantom{-} \text{if}\phantom{-} \nu(t)=i \\
0 \phantom{-} \text{if}\phantom{-} \nu(t)\neq i \\
\end{cases} \forall i\in\I,  t\geq 0.
\end{equation}
Note that the sum of all activation times equals the total elapsed time
\begin{equation}
\sum_{i=1}^{ N} t_a(i,t,t_0) 
=t-t_0,  \quad \forall t\geq t_0.  \label{time_1}   
\end{equation}
For a sequence $\mu$ of nonswitching impulse times, let $n(i,t,t_0)$ denote the number of impulses within mode $i\in\I$ that occur in the interval $(t_0,t]$
\begin{align}
  n(i,t,t_0) &:= \# \lp \mu_{i}  \cap (t_0,t]   \rp,\\
 \mu_i &:= \lp \mu \cap \{t: \nu(t)=i\} \rp %
 \label{def:impulse_times}, \quad \forall i \in \I.
\end{align}
Thus, $\sum_{i=1}^{ N} n(i,t,t_0)=n_\mu(t,t_0)$  for all $t\geq t_0$.

\jlma{We assume that each subsystem has an upper or lower bound for the number of impulse jumps that may occur without switching, depending on whether they are 
regarded
as stabilizing or destabilizing, i.e., depending on whether each $\bar r(i,i)$ for the bound in \eqref{ass03-3} is greater or less than $1$.
This consideration is formalized in the following assumption, where by convention $s/\infty=0$ for all $s\in \R$.
\begin{assumption}[Mode-dependent ADT conditions regarding impulses within each subsystem]
\label{assumption32} $\Sigma_{\Lambda}$ satisfies Assumption \ref{ass031}, and for each $i\in \I$, there exist quantities $T^i_J \in \R_{> 0} \cup \{\infty\}$
and $N_0^i \in \R$ such that
\ajv{
\begin{equation*}
N_0^i \geq 1 \text{~~if~~} \bar r(i,i) \geq 1 \quad \text{and}  \quad N_0^i \leq -1  \text{~~if~~} \bar r(i,i) <1,  
\end{equation*}
}
such that for each $(\nu,\mu)\in\Lambda$, 
for all $i\in \I$, $t_0\geq 0$ and $t\geq t_0$, $n(i,t,t_0)$ satisfies
\begin{subequations}
\label{ass:r<>1}
    \begin{equation}
    n(i,t,t_0)\leq  N_0^i+ \frac{t_a(i,t,t_0)}{T_J^i}, \quad \forall i: \bar r(i,i) \ge 1,\label{n_imp>1} 
    \end{equation}
    \begin{equation}
    n(i,t,t_0)\geq  N_0^i+ \frac{t_a(i,t,t_0)}{ T_J^i}, \quad \forall i: \bar r(i,i)< 1.\label{n_imp<1}
    \end{equation}
\end{subequations}
\end{assumption}
When $\bar r(i,i)>1$, the impulses occurring within an interval where mode $i$ is active have a destabilizing effect, so we need to bound its number from above. On the contrary, when $\bar r(i,i)<1$ the impulses have a stabilizing effect and then we bound its number from below. When $\bar r(i,i)=1$, the effect of the impulses is neutral, but since we are considering strong stability, where the decaying bound also depends on the number of jumps, we need to bound this number from above.
}
\subsection{Improving bounds by increasing the mode sequence length} 
\jlma{When for some $i\neq j$, $\bar r(i,j)>1$, at first glance the switchings seem to have a destabilizing effect, but this is not necessarily true.
}
\blue{In the following, we present an example in which increasing the number of switchings improves convergence.}
Consider a
family $\Sigma_{\Lambda}$ satisfying Assumption \ref{ass031}.
If we bound each $\bar r(\nu(\tau_k^-),\nu(\tau_k))$ in  \eqref{func_mu3}
according to 
the maximum of $\bar r(i,j)$ among all possible jumps between subsystems, i.e. 
\begin{equation}
   R(1):=\max \{\bar r(i,j): (i,j) \in \J \}, \label{R_b}    
\end{equation}
we may obtain an overly conservative state bound.
For illustration, consider a switched impulsive system with $N=3$ subsystems and the
jump graph in Figure~\ref{fig:WDG_example2}.
Consider the following matrix arrangement of weights $\bar r(i,j)$  in \eqref{ass03-3}
with $(i,j)\in \J \cup \bJ$
\begin{equation}
\bar R=
\begin{bmatrix}
    0.216   & 0.003 &  - \\ 
    0.274   & 0.004 &  0.195 \\
    1.656   & - &  0.525 
\end{bmatrix}.   
\end{equation}
The values on the diagonal represent the values of $\bar{r}(i,i)$ for impulse jumps that occur within each mode, while the values outside the diagonal represent the values of $\bar{r}(i,j)$ for jumps between possible mode switchings.
Bounding $\bar r(i,j)$ with $R(1)=1.656$ for every $i\neq j$ in order to bound the trajectories would result in an overly conservative and probably useless bound as far as stability is concerned.
Instead, we propose a better way to compute the bound that includes $R(1)$ as a special case. 
\blue{We define the
maximum combined weight $R(L)$ as the maximum of the product of the weights $\bar r(i,j)$ along all admissible mode sequences of length $L\in \N$ as follows}
\begin{equation}
R(L):=\max  \bigg\{ \prod_{\ell=1}^L \bar r(j_{\ell-1},j_{\ell}): 
 (j_{\ell-1}, j_{\ell})\in \J \medspace \forall \ell \in \{1,2,\hdots, L \} \bigg\}.   \label{Rell} 
\end{equation}
\blue{For the present example, computing the maximum combined weight for a sequence length}
 $L=2$, we obtain $R(2)=\bar r(2,3)  \bar r(3,1) \approx 0.323<1$. 
This implies that while 
\ajv{the effect of the jump introduced by a single switching}
may not ensure a decrease in the Lyapunov function, \ajv{that of the jumps introduced by two successively admissible switchings does indeed ensure such a decrease.}
\jlma{These considerations lead to the following assumption about the frequency of switchings}
\jlma{\begin{assumption}[ADT conditions regarding switchings]
\label{ass:jumps}
Every $(\nu,\mu) \in \Lambda$ 
satisfies a) or b).
\begin{enumerate}[a)]
\item \label{ass:jumps:item:a} There exist quantities 
$\bar N_0\geq 1$ and $\bar T_S>0$  
such that $n_\nu(t,t_0)$ for all $t \geq t_0\ge 0$ satisfies  
\begin{equation}
     n_\nu(t,t_0)\leq  \bar N_0+ \frac{t-t_0}{\bar T_S}.  \label{n_nua}
\end{equation}
\item \label{ass:jumps:item:b}  There exist quantities 
$\munderbar{N}_0\leq -1$ and 
$\munderbar{T}_S \in \R_{>0} \cup\{\infty\}$
such that $n_\nu(t,t_0)$ for all $t \geq t_0\ge 0$ satisfies 
\begin{equation}
    n_\nu(t,t_0) \ge \munderbar{N}_0+ \frac{t-t_0}{ \munderbar{T}_S}.   \label{n_nub} 
\end{equation}
\end{enumerate}
\end{assumption}
}
\jlma{Assumption~\ref{ass:jumps} item~\ref{ass:jumps:item:a}) will be considered when switchings are destabilizing or have a neutral effect, i.e. $R(L)\ge 1$ for every $L$, while item~\ref{ass:jumps:item:b}) will be assumed when they have a stabilizing effect, i.e. $R(L)<1$ for some $L$.}

In order to derive conditions for strong GUES, 
we next introduce coefficients $c_s \geq 0$ that weight how the effect of the switching is computed, either according to the 
elapsed time
or the number of switchings.
To simplify the notation in accordance with Assumption~\ref{ass:jumps}, let $L\in \N$ be given, and define the quantities $T_S$ and $N_S$, 
as follows
\begin{align}
&T_S:=
\begin{cases}
\bar T_S  \phantom{||||} \text{if}\phantom{-} R(L) \geq 1 \\
\munderbar {T}_S \phantom{||||}\text{if}\phantom{-} R(L) < 1
\end{cases}  \qquad 
N_S:=
\begin{cases}
c_s \bar  N_0 \phantom{-||||-} \text{if}\phantom{-} R(L) \geq 1 \\
c_s  \munderbar   {N}_0-L\phantom{||}\text{if}\phantom{-} R(L) < 1
\end{cases}. \label{N0TSnus}   
\end{align}
Additionally, we introduce the coefficients $c_i \geq 0$ for each subsystem to weight the effect of their continuous and impulsive parts.
The quantities in \eqref{eq:alf_nu} and \eqref{eq:alf_i} will be used repeatedly throughout the paper to state our results
\begin{align}
\lambda_s(c_s,L)&:= c_s \frac{\ln(R(L))}{T_S L},  \qquad  \qquad \medspace \medspace  r_s(c_s,L):=   (1-c_s) \frac{\ln(R(L))}{L} \label{eq:alf_nu}, \\
\lambda_i(c_i)&:=  \bar \lambda(i) +  c_i \frac{\ln \lp \bar r(i,i) \rp}{ T_J^i},  \quad \quad \medspace\medspace \medspace\medspace  r_i(c_i):=(1-c_i)  \ln \lp \bar r(i,i) \rp.    \label{eq:alf_i}  
\end{align}
Finally, for the sake of convenience, define
\begin{align}
R(0)&:=1  \quad \text{and } \quad
\hat R(L):= \max_{0\leq \ell \leq L-1} R(\ell) \label{hatR}.
\end{align}
We are ready to use the previous definitions to present our first result.
Theorem~\ref{thm32_strong2} presents 
sufficient conditions
that ensure strong GUES for families of systems.
\begin{theorem}
\label{thm32_strong2}
Suppose that the family of switched impulsive systems $\Sigma_\Lambda$ satisfies assumptions~\ref{ass031} and~\ref{assumption32}.
Select $L \in \N$ and compute $R(L)$ 
and $\hat R$ as in \eqref{Rell} and \eqref{hatR}, respectively. Suppose that Assumption ~\ref{ass:jumps}~a) holds if $R(L)\ge 1$ and that Assumption~\ref{ass:jumps}~b) holds when $R(L)<1$.  
Select $c_s \geq 0$ and compute $T_S$ and $N_S$ as in \eqref{N0TSnus}.
Select $c_i\geq 0$ for all $i\in \I$ and define
\begin{equation}
C_1:= 
\sum_{i=1}^{ N} c_i N_0^i  \ln (\bar r(i,i))+ N_S\frac{\ln(R(L))}{L}+ \ln( \hat R(L) ). \label{H2_S}
\end{equation}
If there exist 
$C_0 \in \R$ and $\lambda_0>0$ 
such that for all $(\nu, \mu) \in   \Lambda$ and  
all $t_0\geq 0$ 
\begin{align}
&\lambda_s(c_s,L) (t-t_0)+r_s(c_s,L) n_\nu(t,t_0) + \sum_{i=1}^{ N}  \lambda_i(c_i) t_a(i,t,t_0) +r_i(c_i) n(i,t,t_0) 
 \es
&\leq C_0- \lambda_0 (t-t_0+n_\mu(t,t_0)+n_\nu(t,t_0)), \quad \forall t\geq t_0, \label{H3}
\end{align}
then 
$\Sigma_{ \Lambda}$
is S-GUES 
with $K$ and $\lambda$ as in \eqref{eq:K_lamb2} defined by $C:=C_0+C_1$ and $\lambda_0$.
\end{theorem}

\begin{proof}
Let $(\nu,\mu)\in \Lambda$.  Define $\hat r_k :=\bar r(\nu(s_k^-),\nu(s_k)) $ where $s_k \in s(\nu)$ for all $k\in \N$ are the
ordered switching times.
Define  $\ell := \min \{i \in \N: s_i \in s(\nu), s_i >t_0 \}$.
Thus, from the definitions of $\lambda(t)$ and $r(t)$ in \eqref{func_lambda3} and \eqref{func_mu3}, for all $t\geq t_0 \geq 0$ we have 
\begin{equation}
\int_{t_0^+}^t  \lambda(\tau)+ r(\tau) d \tau=\int_{t_0^+}^t  \lambda(\tau) d \tau+\int_{t_0^+}^t r(\tau) d \tau  
=  \sum_{i=1}^{ N}  \lp \bar \lambda(i) t_a(i,t,t_0)+    \ln \lp \bar r(i,i) \rp n(i,t,t_0) \rp +  \sum_{k=\ell}^{ n_\nu(t,t_0)+\ell-1}  \ln \lp \hat r_k \rp. \label{d1-10}
\end{equation}
Note that, for a given  $\nu$, 
the number of completed \blue{admissible mode sequences of length}  $L$ 
occurring during the interval $t-t_0$
is given by $\left\lfloor \frac{n_\nu(t,t_0)}{L} \right\rfloor$ while
$mod (n_v(t,t_0),L) \leq L-1$ determines the remaining  number of switchings in that interval.
The last addend  of \eqref{d1-10} can be bounded for all $t\geq t_0$ 
as
\begin{align}
\sum_{k=\ell}^{ n_\nu(t,t_0)+\ell-1}  \ln \lp \hat r_k \rp &= \ln \lp \prod_{k=\ell}^{ n_\nu(t,t_0)+\ell-1}   \hat  r_k \rp   
=\ln \lp   \overbrace{\hat  r_\ell \hat  r_{\ell+1} \hdots \hat r_{n_\nu(t,t_0)+\ell-1} }^{n_\nu(t,t_0) \medspace \text{factors}}\rp   
\es
&\leq \ln \lp R(mod{(n_\nu(t,t_0),L)}) R(L)^{\left\lfloor \frac{n_\nu(t,t_0)}{L}\right \rfloor} \rp \es
&= \ln \lp R(L) \rp  \left\lfloor  \frac{n_\nu(t,t_0)}{L}\right \rfloor + \ln \lp   R(mod{(n_\nu(t,t_0),L)}) \rp 
\leq \ln \lp R(L) \rp \left\lfloor  \frac{n_\nu(t,t_0)}{L}\right \rfloor  + \ln ( \hat R ).
\end{align}
Given $c_s \geq 0$.
If $R(L)\geq 1$ then $\ln(R(L))\geq 0$. Using Assumption~\ref{ass:jumps} item a) we can bound
\begin{align}
\ln \lp R(L) \rp  \left\lfloor  \frac{n_\nu(t,t_0)}{L}  \right\rfloor &\leq \frac{\ln (R(L))}{L}  n_\nu(t,t_0)  =\frac{\ln (R(L))}{L} \lp n_\nu(t,t_0)  (1-c_s) + n_\nu(t,t_0)  c_s \rp \notag \\
&\leq \frac{\ln (R(L))}{L} \lp n_\nu(t,t_0)  (1-c_s)+ \lp \bar N_0 + \frac{t-t_0}{\bar T_S} \rp   c_s  \rp \notag \\
&= \frac{\ln (R(L))}{L} \lp  c_s  \bar N_0 +(1-c_s) n_\nu(t,t_0)   +   \frac{c_s}{ \bar T_S} (t-t_0)    \rp.
\label{cota1-w}
\end{align}
If $R(L)< 1$ then $\ln(R(L))< 0$. Using Assumption~\ref{ass:jumps} item b), we obtain
\begin{align}
\ln (R(L))  \left\lfloor  \frac{n_\nu(t,t_0)}{L}  \right\rfloor &\leq  \ln(R(L)) \lp  \frac{n_\nu(t,t_0)}{L}  -1 \rp   =  \frac{\ln(R(L))}{L} \lp  n_\nu(t,t_0) (1-c_s)+n_\nu(t,t_0) c_s  -L \rp  \notag \\
&\leq \frac{\ln(R(L))}{L} \lp  n_\nu(t,t_0) (1-c_s)+\lp{\munderbar N}_0 + \frac{t-t_0}{{\munderbar T}_\nu} \rp c_s  -L \rp  \notag \\
&=\frac{\ln(R(L))}{L} \lp c_v \munderbar N_0-L + (1-c_s) n_\nu(t,t_0) + \frac{c_s}{ \munderbar T_S} (t-t_0)  \rp.
\label{cota2-w}
\end{align}
According to \eqref{cota1-w}, \eqref{cota2-w} and the definitions of $T_S$ and $N_S$ in \eqref{N0TSnus} we have
  \begin{align}
 \sum_{k=\ell}^{ n_\nu(t,t_0)+\ell-1}  \ln \lp \bar r_k \rp   & \leq \frac{\ln(R(L))}{L} \lp  (1-c_s) n_\nu(t,t_0) + \frac{c_s}{ T_S} (t-t_0)  \rp  
 +  N_S \frac{\ln(R(L))}{L}+ \ln ( \hat R )  \label{d1-w} \\
 &=  r_s(c_s,L) n_\nu(t,t_0) + \lambda_s(c_s,L) (t-t_0)   +  N_S \frac{\ln(R(L))}{L}+ \ln ( \hat R ). \notag
\end{align}    
To simplify the notation, we rename $t_a(i)=t_a(i,t,t_0)$ and $n(i)=n(i,t,t_0)$.
By Assumption~\ref{assumption32}, we can bound each term of the first addend in \eqref{d1-10} as
\begin{align}
\bar \lambda(i) t_a(i)+\ln \lp \bar r(i,i) \rp n(i)
 &= \bar \lambda(i) t_a(i)+\ln \lp \bar r(i,i) \rp n(i) c_i+\ln \lp \bar r(i,i) \rp n(i) (1-c_i) \es
&\leq \bar \lambda(i) t_a(i)+\ln \lp \bar r(i,i) \rp  c_i \lp  N_0^i+ \frac{t_a(i,t,t_0)}{ T_J^i}  \rp  +\ln \lp \bar r(i,i) \rp n(i) (1-c_i) \es
&=\ln \lp \bar r(i,i) \rp  (1-c_i) n(i)+\lp \bar \lambda(i) +  c_i \frac{\ln \lp \bar r(i,i) \rp}{ T_J^i}  \rp  t_a(i) 
+N_0^i\ln \lp \bar r(i,i) \rp  c_i  \es
&= r_i(c_i) n(i)+ \lambda_i(c_i) t_a(i) + c_iN_0^i\ln \lp \bar r(i,i) \rp.   \label{eq:64}
\end{align}
By bounding \eqref{d1-10} using the previous bounds 
 and reordering, we obtain 
\begin{align*}
    \int_{t_0^+}^t \lambda(\tau)+r(\tau) d\tau &\leq  
    \overbrace{%
  \sum_{i=1}^{N}  c_iN_0^i\ln \lp \bar r(i,i) \rp
+ N_S \frac{\ln(R(L))}{L}+\ln ( \hat R )}^{C_1} 
   +\sum_{i=1}^{ N} r_i(c_i) n(i,t,t_0)+ \lambda_i(c_i) t_a(i,t,t_0)  \es
    &+r_s(c_s,L) n_\nu(t,t_0) + \lambda_s(c_s,L) (t-t_0),
\end{align*}
thus the result follows from (\ref{H3}) and Lemma~\ref{cor:1}. 
\end{proof}
\jlma{Theorem~\ref{thm32_strong2} %
provides a method to ensure strong 
stability 
for a broad range of switched impulsive systems based on multiple Lyapunov-like functions.
The values of $c_i$ for each $i \in \I$ in \eqref{eq:alf_i} can be employed to assign weights to the continuous and impulsive parts of each subsystem when applying the direct or reverse bounds in Assumption~\ref{assumption32}. Similarly, the value of $c_s$ in \eqref{eq:alf_nu} represents a trade-off between the number of switchings and the continuous-time, as a consequence of the ADT bounds specified in Assumption~\ref{ass:jumps}. Note that when $c_s=1$, all the effects of the switchings are computed in continuous-time as $r_s(1,L)=0$, while for $c_s=0$, the effects are computed according to the number of switchings as $\lambda_s(0,L)=0$. Similar comments apply to the quantities $\lambda_i(c_i)$ and $r_i(c_i)$ for each mode.
The value of $T_S$ in \eqref{N0TSnus} 
depends on whether $R(L)$ is greater or less than one.
Thus, the value of $R(L)$ may result in either
positive or negative terms $\lambda_s(c_s,L)$ and $r(c_s,L)$ in \eqref{H3} depending on the choice of $c_s$. Note that if $\ln(R(L))$ is negative, then an increase in the number of switchings between subsystems $n_\nu$ contributes to stability; therefore, we assume the number of switchings is bounded from below as per Assumption~\ref{ass:jumps}~b). On the contrary, if $\ln(R(L))$ is positive, we need to impose an upper bound on the number of switchings $n_\nu$ in order to ensure stability, therefore we require Assumption~\ref{ass:jumps}~a). Moreover, the decay rate determined by  \eqref{H3}
can be improved by increasing the length of the sequences of admissible modes in the computation of $R(L)$, as shown below.}
\jlma{\begin{lemma}
\label{lemma:L}
Consider  $L_0,M \in \N$ given and define $L:=L_0M$ then
\begin{equation}
\frac{\ln \lp R(L) \rp }{L} \leq  \frac{\ln \lp R(L_0)\rp}{L_0} \leq \ln(R(1)).
\end{equation}
\end{lemma}
\begin{proof}
By the definition of $R(L)$ in \eqref{Rell} we have $R(L) \leq R(L_0)^{M}$. Taking logarithm, multiplying by $L_0$ and operating yields
\begin{equation*}
\frac{\ln \lp R(L) \rp }{L} \leq  \frac{\ln \lp R(L_0)\rp}{L_0}.
\end{equation*}
Replacing $L$ and $L_0$ by, respectively, $L_0$ and $1$ we obtain $\frac{\ln \lp R(L_0) \rp}{L_0}\le \ln \lp R(1) \rp$.
\end{proof}
}
\jlma{Increasing the $L$ considered may improve the decay rate, but it may worsen 
the bound at initial times if $R(L)$ is less than $1$. This is due to 
the growth of the value of $C_1$ in \eqref{H2_S}, which depends on $L$;
note that in this case $N_S$ is defined as $N_S = \munderbar{N}_0 - L \le -1-L$.
However, combining the bounds obtained for different mode sequence lengths $L$ can result in a better overall bound. }

Theorem~\ref{thm32_strong2} also covers the case of impulsive systems without switching behavior. For this case, $s(\nu)=\emptyset$, and only one flow equation is considered. Thus, as there are no switchings, we have $L=0$ and by definition $R(0)=1$. Subsequently, we obtain $n_{\nu}(t,t_0)=0$, $\lambda_s(c_s,0)=0$ and $r_s(c_s,0)=0$ which can be substituted into \eqref{H3}. 

For systems in which impulsive behaviour only occurs at switching times,
we define $\mu:=\emptyset$, which implies $n_\mu(t,t_0)\equiv 0$ for all $t\geq t_0 \geq 0$, resulting in $n(t,t_0)=n_\nu(t,t_0)$. Therefore, Assumption~\ref{assumption32} is not required.
This simplifies the results in Theorem~\ref{thm32_strong2}, as stated in Corollary~\ref{thm32_strong_simpl1}. This result also covers the case of switching systems without impulsive behavior at switching instants, by selecting $g_{i,j}=\mathrm{id}$ for all $(i,j)\in \J$.

\begin{corollary}[Of Theorem~\ref{thm32_strong2}]
\label{thm32_strong_simpl1}
Suppose that the family of switched impulsive systems $\Sigma_\Lambda$ satisfies 
$n_\mu(t,t_0) \equiv 0$ for all $t\ge t_0\ge 0$ and all $(\nu,\mu)\in \Lambda$ and that 
assumptions of Theorem~\ref{thm32_strong2} hold. Define $C_1:=N_S\frac{\ln(R(L))}{L}+ \ln( \hat R )$.
If there exist 
$C_0 \in \R$ and $\lambda_0>0$ 
such that for all $(\nu,\mu)\in \Lambda$ and all
$t_0\geq 0$ 
\begin{equation}
\sum_{i=1}^{ N} \bar \lambda(i) t_a(i,t,t_0) 
+ \lambda_s(c_s,L) (t-t_0)+r_s(c_s,L) n_\nu(t,t_0) 
\leq C_0- \lambda_0 (t-t_0+n_\nu(t,t_0)), \quad \forall t\geq t_0, \label{H3_simp}
\end{equation}
then $\Sigma_{ \Lambda}$ 
is S-GUES 
with $K$ and $\lambda$ as in \eqref{eq:K_lamb} defined by $C:=C_0+C_1$ and $\lambda_0$.
\end{corollary}

\begin{proof}
Select $c_i=0$ for all $i\in \I$, then from \eqref{eq:alf_i} and \eqref{H2_S}
we obtain $\lambda_i(c_i)=\bar \lambda(i)$ and $C_1=N_S\frac{\ln(R(L))}{L}+ \ln( \hat R )$, respectively. 
Note that $n_\mu(t,t_0)\equiv 0$ implies that $n(i,t,t_0) \equiv 0$ for all $i\in \I$. Setting $\lambda_i(c_i)=\bar \lambda(i)$ and $n(i,t,t_0)\equiv 0$ in \eqref{H3} we obtain \eqref{H3_simp}.
\end{proof}
\jlma{Theorem~5.1 of \cite{mancilla2020uniform} provides related results to ensure strong stability for a class of switched impulsive systems, but that approach exhibits some differences/drawbacks in comparison to the current results. Specifically: a) the sequence of switching times and that of impulse times are the same in Theorem~5.1 of \cite{mancilla2020uniform}, therefore our results apply to a larger class of systems; b) Theorem~5.1 of \cite{mancilla2020uniform} imposes single coefficients for the sets of stabilizing, destabilizing, and neutral subsystems. The neutral subsystems are those that are not guaranteed to be stabilizing or destabilizing. This may lead to extremely conservative bounds
as, for example, stabilizing subsystems that decay at different
rates, fast or slow in the flow or jump maps, are required to share
the same coefficients in \eqref{ass03-2} and \eqref{ass03-3}, respectively.
In the present Theorem~\ref{thm32_strong2}, the coefficients $\bar \lambda(i)$ and $\bar r(i,i)$ correspond to each subsystem, which allows to compute better bounds.}

\subsection{Activation time conditions for group of modes} 
This section presents conditions 
for ensuring strong GUES %
by imposing bounds on the activation times of each subsystem. 
Specifically, assuming that the constraints of Assumption~\ref{assumption32} are satisfied,
we will show that it is sufficient to assume additionally that an affine function of the total elapsed time $(t-t_0)$  bounds the activation time of destabilizing subsystems.

Consider $\lambda_i(1)= \bar \lambda(i)+\frac{\ln(\bar r(i,i))}{T_J^i}$ for all $i\in \I$  as defined by \eqref{eq:alf_i} with $c_i=1$ for all $i\in \I$, and partition the index set $\I$ as 
$\I=\I^u \cup \I^s$,
with %
\begin{align}
\begin{split}
\I^u&:=\{i \in \I: \lambda_i(1)\geq 0\}, \\
\I^s&:=\{i \in \I: \lambda_i(1)<0\}, \\
\end{split}
\label{eq:sets_times}
\end{align}
Activation-time constraints will be placed on groups of subsystems, so that less stringent conditions can be obtained. The sets $\I^u$ and $\I^s$ 
are therefore partitioned as follows
\begin{align}
&\I^u=\I_1\cup \hdots\cup \I_{N^u}, \notag \\
&\I^s=\I_{N^u+1} \cup \hdots \cup \I_{N^\I}, \notag 
\end{align}
with $N^u$ and $N^s$
the number of subsystem groups within $\I^u$ and $\I^s$, 
respectively, with $N^\I:=N^u+N^s$ as the total number of subsystem groups.
We assume the following constraints on the activation times.
\begin{assumption}[Regarding activation times]
\label{ass_times}
There exist real numbers $N_{a}^j$ and $T_{a}^j$ satisfying
\begin{align}
\label{eq:quantities}
    \begin{split}
    &N_{a}^j \geq 0,\quad   \forall j: \I_j \subset \I^u \cup \I^s, \\
& T_{a}^j \geq 0, \phantom{||||||}  \forall j: \I_j \subset \I^u,   \\
& T_{a}^j \leq 0, \phantom{||||||}  \forall j: \I_j \subset \I^s,
\end{split}
\end{align}
such that for all $t \geq t_0$ and $t_0\geq 0$ 
\begin{subequations}
\label{cond_time_weak}
    \begin{equation}
        \sum_{i\in \I_j}t_a(i,t,t_0)\leq T_{a}^j+N_{a}^j(t-t_0), \quad
    \forall j: \I_j \subset  \I^u, \label{cond_time}\\
    \end{equation}
    \begin{equation}
        \sum_{i\in \I_j} t_a(i,t,t_0)\geq T_{a}^j+N_{a}^j(t-t_0), \quad 
    \forall j: \I_j \subset \I^s. \phantom{,}\label{cond_time2}
    \end{equation}
\end{subequations}    
\end{assumption}
\begin{remark}
\label{remark:1}
\ajv{
Considering the equality expressed in \eqref{time_1}, which states that the sum of all activation times equals the total elapsed time, in certain cases, we may obtain better bounds if the quantities in \eqref{eq:quantities} are chosen to satisfy
\begin{equation}
\sum_{1 \leq j \leq N^\I} T_a^j = 0 \quad \text{and} \quad \sum_{1 \leq j \leq N^\I} N_a^j = 1.
\end{equation}
}
\end{remark}

\begin{remark}
\label{remark_time}
When all subsets $\I_j$ are chosen such that each one contains exactly one subsystem $i\in \I$, mode-dependent conditions are obtained from \eqref{cond_time_weak}.
Namely, for all $i \in \I$, we have
\begin{subequations}
\label{cond_time_weak_mode}
                \begin{equation}
    t_a(i,t,t_0)\leq T_{a}^i+N_{a}^i(t-t_0), \quad
    \forall i:\lambda_i(1)\ge 0, \label{cond_time_mode}
    \end{equation}
    \begin{equation}
    t_a(i,t,t_0)\geq T_{a}^i+N_{a}^i(t-t_0), \quad 
    \forall i: \lambda_i(1)<0. \label{cond_time2_mode}        
    \end{equation}
\end{subequations}
\end{remark}

To define the 
quantities $N_a^j$, $T_a^j$ satisfying \eqref{cond_time_weak} 
in Assumption~\ref{ass_times}, we first need to obtain $\lambda_i(1)$ for all $i\in \I$. The computation of each $\lambda_i(1)$ %
depends on the functions $V_i$ and on the values of $ T_J^i$ of the conditions that bound the frequency of the impulses in Assumption~\ref{assumption32}. Similarly to the handling of the switchings in Assumption~\ref{ass:jumps}, Assumption~\ref{ass_times}  
constrains the activation times of groups of subsystems to preserve GUES.
This assumption is much less restrictive than typical activation-time constraints on individual subsystems as the ones described in Remark~\ref{remark_time}.
To see this, consider an example where the subsystems are grouped into sets
of unstable subsystems, stable subsystems with slow convergence and stable subsystems with fast convergence. To ensure stability of the switched system, it is clear that the stable subsystems should be more active in relation to the unstable ones. Grouping the subsystems allows to place activation-time constraints on each group, irrespective of the individual activation times of each subsystem in the group.
For instance, by limiting the activation times of the unstable and the stable but slow groups in relation to that of the stable and fast, not only stability of the switched system can be ensured but also better estimates of the maximum overshoot and decay rate can be obtained. 

\subsection{Sufficient conditions for S-GUES} 
We are almost ready to combine all the presented assumptions and results to establish sufficient conditions for S-GUES, but first, we need the following Lemma~\ref{lemma:r_i_SGUAS},
\ajv{whose proof is provided in Appendix~\ref{appendix:2}}.
It shows that, under the mode-dependent conditions of Assumption~\ref{assumption32},
is possible
\ajv{to select positive quantities $c_i > 0$ in order to derive adequate bounds.} 

\begin{lemma}
\label{lemma:r_i_SGUAS}
Suppose that Assumption~\ref{assumption32} holds. 
Consider $r_i(c_i)$ and $\lambda_i(c_i)$ from \eqref{eq:alf_i}  for all $i\in \I$.
Then, for all $i\in \I$ such that $\bar r(i,i)\neq 1$, there exist  $c_i >0$  so that
\begin{enumerate}[i)]
\item  $r_i(c_i)<0$ and $\lambda_i(c_i)<0$ for all $ \lambda_i(1)<0$, 
\item $r_i(c_i)<0$ and $\lambda_i(c_i)>0$ for all  $ \lambda_i(1)\geq 0$. 
\end{enumerate}
\end{lemma}
Table~\ref{table5} shows how to select appropriate values of  $c_i$ such that 
$r_i(c_i)<0$, according to the proof of Lemma~\ref{lemma:r_i_SGUAS}. This, in turn, may increase the value of $\lambda_i(c_i)$, which bounds the continuous-time evolution of the subsystem. 
By increasing $c_i$, while satisfying the restrictions in Table~\ref{table5}, we also increase the value of $\lambda_i(c_i)$, such that it remains negative in the case of the first column, corresponding to i) of Lemma~\ref{lemma:r_i_SGUAS}, or positive in the case of the second column, corresponding to ii) of Lemma~\ref{lemma:r_i_SGUAS}.
This can be observed in Figure~3.1 which illustrates the results of Lemma~\ref{lemma:r_i_SGUAS}.

\renewcommand{\arraystretch}{1.5}

\begin{table}[ht!]
    \centering
    \begin{tabular}{|l|l|l|}
     \hline 
&$ \lambda_i(1) < 0 $ & $   \lambda_i(1) \geq 0  $  \\ \hline
     $\ln \lp \bar r(i,i) \rp  < 0,  \bar \lambda(i) \leq 0$& $c_i \in(0,1)$& -\\ 
     \hline 
     $\ln \lp \bar r(i,i) \rp  < 0,  \bar \lambda(i)>0$& $c_i \in\left (- \frac{\bar \lambda(i) T_J^i}{\ln(\bar r(i,i))},1\right)$& $c_i \in(0,1)$ \\ \hline 
     $\ln(\bar r(i,i))>0, \bar \lambda(i)<0$& $c_i \in\left(1,- \frac{\bar \lambda(i) T_J^i}{\ln(\bar r(i,i))}\right)$& $c_i \in(1,\infty)$ \\ \hline
    $\ln(\bar r(i,i))>0, \bar \lambda(i) \geq 0$& -& $c_i \in(1,\infty)$ \\ \hline
        $\ln(\bar r(i,i))= 0$& $\nexists$ & $\nexists$ \\ \hline
    \end{tabular}
    \caption{Admissible values of $c_i$ such that $r_i(c_i)=(1-c_i)  \ln \lp \bar r(i,i)\rp<0$ for all possible values of $\ln(\bar r(i,i))$ and $\bar \lambda(i)$ according to Lemma~\ref{lemma:r_i_SGUAS}.}
    \label{table5}
\end{table}
\usetikzlibrary[patterns] 

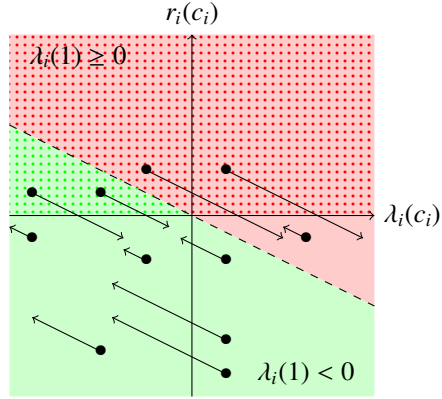
\begin{figure}[ht!]
\label{figure:3}
\begin{center}
\begin{tikzpicture}[scale=0.6]
  \fill[red!20] (0,0) -- (-4,2) -- (-4,4) -- (4,4) -- (4,-2) -- cycle;
  \fill[pattern=dots, , pattern color=red] (0,0) -- (-4,2) -- (-4,4) -- (4,4) -- (4,0) -- cycle;
  \fill[green!20] (0,0) -- (-4,2) -- (-4,-4) -- (4,-4) -- (4,-2) -- cycle;
  \fill[pattern=dots,  pattern color=green] (0,0) -- (-4,0) -- (-4,2) -- cycle;

  \draw[->] (-4,0) -- (4,0) node[right] {$\lambda_i(c_i)$};
  \draw[->] (0,-4) -- (0,4) node[above] {$r_i(c_i)$};
  
  \draw[dashed,black] (-4,2) -- (4,-2);

  \node[black] at (2.5,-3.5) {$\lambda_i(1) < 0$};
  \node[black] at (-2.5,3.5) {$\lambda_i(1) \geq 0$};

  \node at (-2,0.5) {\textbullet};
  \draw[->] (-2,0.5) -- (-2+2*0.75,0.5-1*0.75);
  \node at (-2,-3) {\textbullet};
  \draw[->] (-2,-3) -- (-2-2*0.75,-3+1*0.75);
  
    \node at (-3.5,0.5) {\textbullet};
  \draw[->] (-3.5,0.5) -- (-3.5+2*1,0.5-1*1);
  
    \node at (0.75,-3.5) {\textbullet};
  \draw[->] (0.75,-3.5) -- (0.75-2*1.25,-3.5+1*1.25);
    \node at (0.75,-2.75) {\textbullet};
  \draw[->] (0.75,-2.75) -- (0.75-2*1.25,-2.75+1*1.25);

    \node at (0.75,-1) {\textbullet};
  \draw[->] (0.75,-1) -- (0.75-2*0.5,-1+1*0.5);
      \node at (0.75,1) {\textbullet};
  \draw[->] (0.75,1) -- (0.75+2*1.5,+1-1*1.5);
      \node at (-1,1) {\textbullet};
  \draw[->] (-1,1) -- (-1+2*1.5,+1-1*1.5);
        \node at (-1,-1) {\textbullet};
  \draw[->] (-1,-1) -- (-1-2*0.25,-1+1*0.25);
      \node at (-3.5,-0.5) {\textbullet};
  \draw[->] (-3.5,-0.5) -- (-3.5-2*0.25,-0.5+1*0.25);
        \node at (2.5,-0.5) {\textbullet};
  \draw[->] (2.5,-0.5) -- (2.5-2*0.25,-0.5+1*0.25);
\end{tikzpicture}
\end{center}
\caption{Upper half-plane (dotted pattern): points where \eqref{n_imp>1} holds. Lower half-plane: points where \eqref{n_imp<1} holds.
For different initial values $\lambda_i(0)=\bar \lambda(i)$ and $r_i(0)=\ln(\bar r(i,i))$, the values of $\lambda_i(c_i)$ and $r_i(c_i)$ change in the direction of the arrows as $c_i \geq 0$ increases. %
Green: points that satisfy $\lambda_i(1)<0$. These allow increasing $c_i \geq 0$ in order to %
make both $r_i(c_i)<0$ and $\lambda_i(c_i)<0$.
 Red: points for which $\lambda_i(1)\geq0$. These allow increasing $c_i\ge 0$ in order to make $r_i(c_i)<0$ at the expense of making $\lambda_i(c_i)$ positive and larger. The same value $T_J^i$ of the corresponding direct or reverse constraint of Assumption~\ref{assumption32} is considered for each case in this example, thus the segments have the same slope. 
}
\end{figure}
Lemma~\ref{lemma:r_i_SGUAS} is used 
in 
Theorem~\ref{thm:strong_GUES}, which 
presents conditions that ensure S-GUES for a given family of systems. 
\begin{theorem}
\label{thm:strong_GUES}
Suppose that the family of switched impulsive systems $\Sigma_\Lambda$ satisfies assumptions~\ref{ass031},  \ref{assumption32} and
\ref{ass_times}. Select $L\in \N$ and compute $R(L)$ and $\hat R$ as in \eqref{Rell} and \eqref{hatR}, respectively. Suppose Assumption~\ref{ass:jumps}~a) holds if $R(L)\ge 1$ and that Assumption~\ref{ass:jumps}~b) holds if $R(L)< 1$. 
Select $c_i \geq 0$ such that $r_i(c_i)<0$ for all $i\in \I$ such that $\bar r(i,i) \neq 1$
and $c_s \geq 0$ such that $r_s(c_s,L)<0$.
Compute $T_S$, $N_S$ as in \eqref{N0TSnus} and 
define
\begin{align}
C&:= 
\sum_{j=1}^{N^\I}   T_a^j \max_{i\in \I_j} \{\lambda_i(c_i)\}
+
\sum_{i=1}^{ N}   c_i N_0^i  \ln (\bar r(i,i)) + N_S\frac{\ln(R(L))}{L}+ \ln( \hat R(L) ),  \label{eq:C_thmS_GUES}\\
\lambda_J&:= \sum_{j=1}^{N^\I} N_a^j \max_{i\in \I_j} \{\lambda_i(c_i)\},
\quad r_J:= \max \{r_i(c_i) : \bar r(i,i) \neq 1\},
\label{eq:lambda_mu-r_mu}\\
\lambda_0 &:=-\max \left \{\lambda_J +\lambda_s(c_s,L),r_s(c_s,L),r_J\right\}. \label{eq:lam_0}
\end{align}
If $\lambda_0>0$, then  the family 
is S-GUES with $K$ and $\lambda$ as in \eqref{eq:K_lamb2}.
\end{theorem}

\begin{proof}
Assumptions of 
Theorem~\ref{thm32_strong2} hold.
Rename $t_a(i):=t_a(i,t,t_0)$, $n(i):=n(i,t,t_0)$, $n_\mu:=n_\mu(t,t_0)$, $n_\nu:=n_\nu(t,t_0)$
and $N^\I=N^u+N^s+N^0$ for the sake of simplicity.
By Assumption~\ref{ass_times} we can bound 
\begin{align*} 
&\sum_{i=1}^{ N}  r_i(c_i) n(i) + \lambda_i(c_i) t_a(i) \leq \max \{r_i(c_i): r_i(c_i)\neq 0 \}  \sum_{i=1}^{N} n(i) \notag + \sum_{j=1}^{N^\I} \max_{i\in \I_j} \{\lambda_i(c_i)\}  (T_a^j+N_a^j(t-t_0)) \es
&\leq r_J n_\mu + \sum_{j=1}^{N^\I} T_a^j \max_{i\in \I_j} \{\lambda_i(c_i)\}  + (t-t_0)  
 \sum_{j=1}^{N^\I} N_a^j \max_{i\in \I_j} \{\lambda_i(c_i)\}  
= r_J n_\mu + \sum_{j=1}^{N^\I} T_a^j \max_{i\in \I_j} \{\lambda_i(c_i)\}  +  
 \lambda_J (t-t_0).
\end{align*}
Define $C_0:=\sum_{j=1}^{N^\I} T_a^j \max_{i\in \I_j} \{\lambda_i(c_i)\}  $
Thus, the term in the left of \eqref{H3} can be bounded by
\begin{align}
 C_0 
  +  \lambda_J (t-t_0)  + \lambda_s( t-t_0) + r_J n_\mu + r_s n_\nu 
 &\leq  C_0 + (\lambda_J +\lambda_s) (t-t_0)  + \max\{ r_J,r_s \}( n_\mu+ n_\nu)  \es
 & \leq C_0 - \lambda_0 (t-t_0+ n_\mu+ n_\nu). \notag
\end{align}
\vspace{-0.15cm}
\end{proof}
There are various degrees of freedom at our disposal that provide flexibility to ensure 
S-GUES under a wide range of scenarios.  
For the switching part, we need to determine the 
mode sequence length
$L$ and the values of $\bar N_0$ and $\bar T_S$, or alternatively, $\munderbar{N}_0$ and $\munderbar{T}_S$, which determine the bounds on the number of switchings, as specified in Assumption~\ref{ass:jumps}.
For each individual subsystem, the values $N_0^i$ and $T_J^i$ are used to bound the impulsive jumps of each subsystem as defined in Assumption~\ref{assumption32} and $N_a^j$ and $T_a^j$ to bound the activation time of groups of subsystems (or individual subsystems) as presented in Assumption~\ref{ass_times}.
Furthermore, once we have established these bounds on the switchings and impulse jumps of each subsystem, we can choose suitable values of $c_s$ and $c_i$ for all $i\in \I$ to balance the contributions of the continuous and impulsive parts of each subsystem. 

For systems where impulsive behaviour only occurs at switching times, i.e. $\mu=\emptyset$ and $n_\mu(t,t_0) \equiv 0$, there exist simpler conditions based on Corollary~\ref{thm32_strong_simpl1}, which we present below.

\begin{theorem}
\label{cor:SGUES_simp}
Suppose that the family $\Sigma_\Lambda$ satisfies  
$n_\mu(t,t_0) \equiv 0$ and
assumptions~
\ref{ass031} and \ref{ass_times}. 
Select $L\in \N$
and compute $R(L)$ and $\hat R$ as in \eqref{Rell} and \eqref{hatR}, respectively.
Suppose that Assumption \ref{ass:jumps}~a) holds if $R(L)\ge 1$ and that Assumption \ref{ass:jumps}~b) holds if $R(L)< 1$. Compute $T_S$ and $N_S$ as in \eqref{N0TSnus} and define
\begin{align}
 C 
&:=
  \sum_{j=1}^{N^\I} T_a^j \max_{i\in \I_j} \{\bar \lambda(i)\}
  + N_S\frac{\ln(R(L))}{L}+ \ln( \hat R(L) ), \label{H2-simp}\\
   \lambda_0 &:=  -\max \left \{
  \sum_{j=1}^{N^\I} N_a^j \max_{i\in \I_j} \{\bar \lambda(i)\}
  + c_s \frac{\ln(R(L))}{T_S L},  (1-c_s) \frac{\ln(R(L))}{L}\right \}.\label{H1-simp}
\end{align}
If there exist $c_s \in [0,1)$ such that $\lambda_0>0$, then  the family 
is S-GUES with $K$ and $\lambda$ as in \eqref{eq:K_lamb2}.
\end{theorem}
\begin{proof}
Note that assumptions of Corollary~\ref{thm32_strong_simpl1} hold. 
We have $\hat \lambda(i)=\bar \lambda(i)$ for all $i\in\I$.
From \eqref{H3_simp}, using Assumption~\ref{ass_times} we have
\begin{align*}
&\sum_{i=1}^{ N} \bar \lambda(i) t_a(i,t,t_0) 
+ \lambda_s(c_s,L) (t-t_0)+  r_s(c_s,L) n_\nu(t,t_0)   \es
&< \sum_{j=1}^{N^\I} T_a^j \max_{i\in \I_j} \{\bar \lambda(i)\} + \lp  \sum_{j=1}^{N^\I} N_a^j \max_{i\in \I_j} \{\bar \lambda(i)\} 
+ \lambda_s(c_v,L) \rp (t-t_0) +r_s(c_s,L)  n_\nu(t,t_0) \es
&< \sum_{j=1}^{N^\I} T_a^j \max_{i\in \I_j} \{\bar \lambda(i)\} + \max \left \{  \sum_{j=1}^{N^\I} N_a^j \max_{i\in \I_j} \{\bar \lambda(i)\} 
+ \lambda_s(c_v,L),r_s(c_s,L) \right \} ((t-t_0) + n_\nu(t,t_0)) \es
&=\sum_{j=1}^{N^\I} T_a^j \max_{i\in \I_j} \{\bar \lambda(i)\}-
\lambda_0 ((t-t_0)+n_\nu(t,t_0)). 
\end{align*} 
\end{proof}
For certain cases, it is possible to achieve S-GUES even if the continuous-time part of the subsystems has a destabilizing effect
determined by their activation time. If the value %
$   \sum_{j=1}^{N^\I} N_a^j \max_{i\in \I_j} \{\bar \lambda_i\} $
in \eqref{H1-simp}
is positive and the switching between subsystems is stabilizing, we can ensure S-GUES if
the overall sum in the first term of \eqref{H1-simp} can be made negative by appropriately selecting the values of $c_s$ and $T_S$. 

\section{Examples}
\label{sec:example}

\subsection{S-GUES family with an unstable subsystem}
We provide a numerical example to illustrate the applicability of our results. This example 
shows that increasing the switching sequence length $L$ in the computation of the bound allows us to prove strong stability. Additionally, we show that by adequately selecting the coefficients $c_s$ and $c_i$ for each mode, it is possible to obtain tighter bounds.
Consider a family 
\ajv{of switched impulsive systems of the form 
\eqref{impulsive_new} with linear flow and jump maps $f_i(x)=A_i x$ and $g_{i,j}(x)=J_{i,j}x$,}
with $N=2$ subsystems
 \blue{having the jump graph in Figure~\ref{fig:WDG_example3}} 
and
\begin{equation}
\label{eq:matrices}
A_1=
\begin{bmatrix}
-1.4 & 0.6 \\
-0.5 & -0.3 
\end{bmatrix},  \quad  
A_2=
\begin{bmatrix}
4 & 3 \\
-1 & 2 
\end{bmatrix}, \quad
\blue{J_{2,1}}=J_{1,1}=
\begin{bmatrix}
0.105  &       0 \\
0 &    0.11
\end{bmatrix}, \quad
\blue{J_{1,2}}=J_{2,2}=
\begin{bmatrix}
1.26 & 0 \\
0 & 1.26
\end{bmatrix}.
\end{equation}
\vspace{-0.5cm}
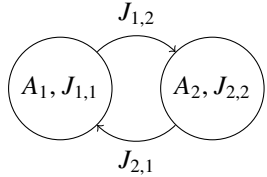
\begin{figure}[ht!]
\begin{center}
\begin{tikzpicture}
  \node[draw, circle] (1) at (0,0) {$A_1,J_{1,1}$};
  \node[draw, circle] (2) at (2,0) {$A_2,J_{2,2}$};
  \draw[->] (1) edge[bend right=-45] node[above] {$J_{1,2}$} (2);
  \draw[->] (2) edge[bend right=-45] node[below] {$J_{2,1}$} (1);
\end{tikzpicture}
\end{center}
\vspace{-0.5cm}
\caption{The jump graph of possible jumps for the system.}
\label{fig:WDG_example3}
\end{figure}

Consider that the constraints in assumptions~\ref{assumption32} and~\ref{ass:jumps} 
hold with
\begin{align}
\label{Quantities}
\begin{split}
&\bar N_0=1, \phantom{11||||} \munderbar{N}_0=-1, \phantom{11}  N_0^1=-1, \phantom{1111}  N_0^2=1,  \\
&\bar T_S=0.1, \phantom{11} \munderbar{T}_S=0.1, \phantom{11}  T_J^1=0.085,\phantom{11}   T_J^2=0.024.
\end{split}
\end{align}
Note that matrices $A_1$ and $J_{1,1}$ are Hurwitz and Schur stable respectively, while
$A_2$ and $J_{2,2}$ are not.
This implies that even if all conditions of assumptions~\ref{assumption32} and~\ref{ass:jumps} are satisfied, the unstable subsystem given by $(A_2,J_{2,2})$ can still be active for most of the time. 
Moreover, there exists a destabilizing impulse during the jump from subsystem $1$ to $2$.
Let $\Sigma_{\Lambda}$ be the family of systems  
which additionally satisfies the following activation time conditions
\begin{subequations}
\label{eq:c}
\begin{eqnarray}
&t_a(1,t,t_0) \geq T_{a}^1+N_a^1(t-t_0), \label{tiempo1} \\
&t_a(2,t,t_0) \leq T_{a}^2+N_a^2(t-t_0), \label{tiempo2}
\end{eqnarray}
\end{subequations}
with
$N_{a}^2=0.56$,
$T_{a}^2=0.03$, $N_a^1=(1-N_a^2)=0.44$ and $T_a^1=-T^2_a=-0.03$.

To prove that the family $\Sigma_{ \Lambda}$
is S-GUES, we will compute the necessary constants constants $ K$ and $\lambda$ outlined in  
Theorem~\ref{thm:strong_GUES}.
\ajv{For the computation of the functions $V_i$ in Assumption~\ref{ass031}, we will follow the results outlined in Appendix~\ref{app:boundingLSI}.
}
First, define $\I:=\{1,2\}$ and the index sets $\I_d=\{1\}$, $\I_u=\{2\}$.
Select $Q_1:=I$ and $P_2=I$,
and apply Lemma~\ref{property3} to obtain
the following quantities specified by  \eqref{lyap4_c}, \eqref{tildeQ}, \eqref{prop2_2} and \eqref{prop2_3} 
\begin{align*}
&P_1=
\begin{bmatrix}
1.0111  &  0 \\
0  &  1.0122
\end{bmatrix}, \quad
\tilde Q_1=
\begin{bmatrix}
2.8312  &      -0.1006 \\
-0.1006   & 0.6073
\end{bmatrix}, \quad 
\tilde Q_2=
\begin{bmatrix}
-8  &      -2 \\
-2   & -4
\end{bmatrix}, \quad 
\bar R=
\begin{bmatrix}
 0.0121 
  &  1.5701 \\
0.0122  &  1.5876
\end{bmatrix}, \notag \\
&\bar \lambda(1)=-0.5955,  \medspace \medspace  \bar \lambda(2)=8.8284.
\end{align*}
Thus, Assumption~\ref{ass031} is satisfied 
with the previous values of $\bar \lambda(1)$, $\bar \lambda(2)$, $\bar r(i,j) = \bar R_{i,j}$ and $V_i(x)=x^T P_i x$ for all $i\in \I$.
By computing %
$ \lambda_i(1):= \bar \lambda(i)+\frac{\ln(\bar r(i,i))}{T_J^i}$ for each mode, we obtain 
$ \lambda_1(1) = -52.5314$
and  $\lambda_2(1)=   28.0877  
$.
The signs of $\lambda_1(1)$ and $\lambda_2(1)$ and the constraints \eqref{eq:c} imply that Assumption~\ref{ass_times} is satisfied with $\I^u=\I_1=\{2\}$ and $\I^s=\I_2=\{1\}$.  As a result, all the conditions 
of Theorem~\ref{thm:strong_GUES} are met.

To apply Theorem~\ref{thm:strong_GUES},
note that $\ln (\bar r(1,1))<0$, $\bar\lambda(1)\leq 0$ and
$\ln (\bar r(2,2))>0$, $\bar\lambda(2)\geq 0$ then,
according to Table~\ref{table5}, we can choose
 $c_1 \in (0,1)$ and  $c_2 \in (1,\infty)$ 
such that $r_1(c_1)<0$ and $r_2(c_2)<0$.
Select $c_1:=0.8$ and $c_2:=2.3$, to obtain $r_J=-0.6009$ and $\lambda_J=11.2065$
from \eqref{eq:lambda_mu-r_mu}.
To show the improvement of the bound, we will 
use different 
lengths $L \in \{1,2,3\}$ for the mode sequences.
Select $c_s=0.6$, by which we obtain the values displayed in Table~\ref{table10},
\begin{table}[h]
    \centering
    \begin{tabular}{|c|c|c|c|c|}
    \hline
     $L$  & 1 & 2 & 3 \\
     \hline
        $R(L)$ &
    1.5701   & 0.0192    &0.0302  \\ \hline
      $\lambda_s(0,L)$ &  2.7068  &-11.8537    &-7.0002  \\ \hline
    $r_s(0,L)$ & 0.1805  &-0.7902    &-0.4667
\\ \hline
    $\lambda(L)$ & -6.9566   & 0.3004  &  -2.1031
\\ \hline
    $K(L)$  &47.839 & 682.917  &427.5748
    \\ \hline
    \end{tabular}
    \caption{Values of $R$, $C$,  $K$ and $-\lambda$  for $L \in \{1,2,3\}$.}
    \label{table10}
\end{table}
where $C(L)$ and $\lambda_0(L)$ in \eqref{eq:C_thmS_GUES} and \eqref{eq:lam_0}, define the bound
\begin{align}
|x(t)|\leq \beta_L(|x_0|,t-t_0+n(t,t_0)):= K(L) e^{-\lambda(L) (t-t_0+n(t,t_0))} |x_0| \quad \forall t\geq t_0.
\label{functionsL}
\end{align}
with unique $K=\exp \lp\frac{C}{2} \rp \lp\frac{\max_{i\in \I}\lambda_M(P_{i})}{\min_{i\in\I}\lambda_m(P_{i})} \rp^{1/2}$  and $\lambda= \lambda_0/2$ as in \eqref{eq:K_lamb2}.
For $L=1$, we obtain $R(1)\geq 1$,
which leads to a positive value of $-\lambda(1)$ and a divergent exponential bound. 
However, for $L=2$,
we obtain a negative value for $-\lambda(2)$, 
and hence, Theorem~\ref{thm:strong_GUES}
ensures S-GUES of $\Sigma_{  \Lambda}$ without any modifications to the $V_i$ functions used. 
For 
$L=3$, although the bound is again divergent, it is better than for 
$L=1$, as per Lemma~\ref{lemma:L}. 

We next simulate a system from the presented S-GUES family $\Sigma_{ \Lambda}$
using various initial conditions satisfying $|x_0|=1$ and $t_0=0$.
For the current simulation, the
required constraints 
in assumptions~\ref{assumption32} and~\ref{ass:jumps} and the activation-time constraints \eqref{eq:c} 
in Assumption~\ref{ass_times} are satisfied, as
shown in Figure~\ref{figure1}
and Figure~\ref{fig:times}, respectively.
Figure~\ref{figure2_threebounds} shows the three bounds corresponding to the different set of values of $L$, $c_1$, $c_2$ and $c_s$. Specifically,  the third bound in red solid line corresponds to the S-GUES bound specified by the values of Table~\ref{table10} with L=2.
Note that the bound that results in the minimum maximum overshoot 
$K$ (in black dashed line) diverges, while the other two  ensure stability. Furthermore, the one in blue dotted line exhibits a slow decay rate compared to the one in red solid line.
By combining these bounds we guarantee 
\ajv{strong global uniform asymptotic stability}
for $\Sigma_{ \Lambda}$, with a bound of the form  $|x(t)|\leq\beta(|x_0|, t-t_0+n(t,t_0))$, 
with a function $\beta \in \KL$ defined as the minimum of all functions defined as in \eqref{functionsL} for the three set of constants $L$, $c_1$, $c_2$ and $c_s$ used; i.e. $\beta:= \min\{\beta_1,\beta_2,\beta_3\}$. 
As a result, this approach allows us to benefit from both a low maximum overshoot and a fast decay rate.
\vspace{-0.2cm}
\begin{figure}[ht!]
\centering
\includegraphics[width=0.5\columnwidth]{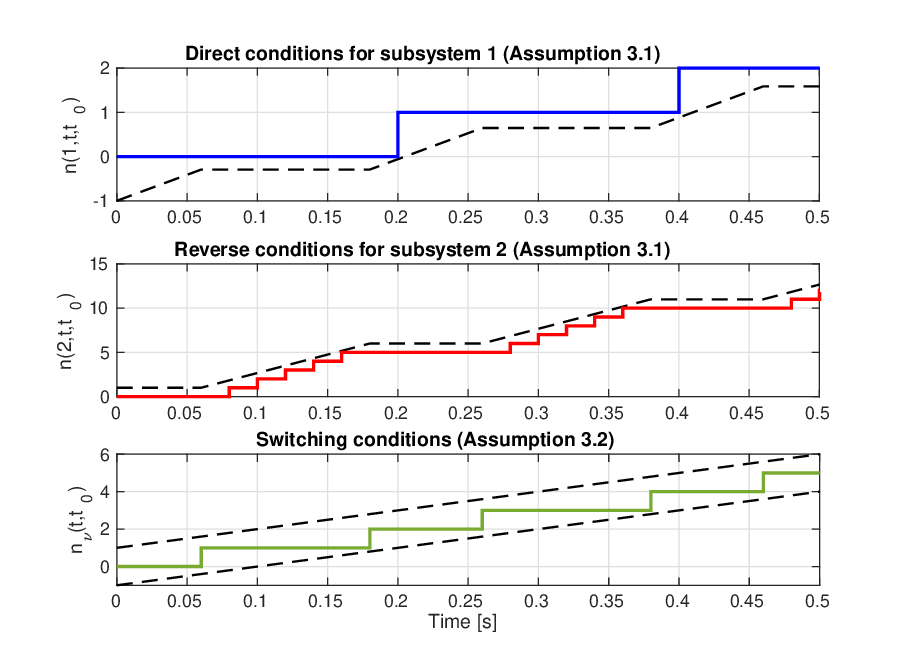}
\vspace{-0.5cm}
\caption{ %
Bounds from assumptions~\ref{assumption32}  and~\ref{ass:jumps} for the current simulation with $t_0=0$. 
Top plot: the number of switchings $n_\nu(t,t_0)$ is bounded from above or below according to the value of $R(L)$. Middle and bottom: the impulsive jumps of subsystems $1$ and $2$ are upper and lower bounded according to \eqref{n_imp>1} and \eqref{n_imp<1}, respectively. All plots: bounds are plotted as dashed lines.}
\label{figure1}
\end{figure}
\begin{figure}[ht!]
\centering
\includegraphics[width=0.5\columnwidth]{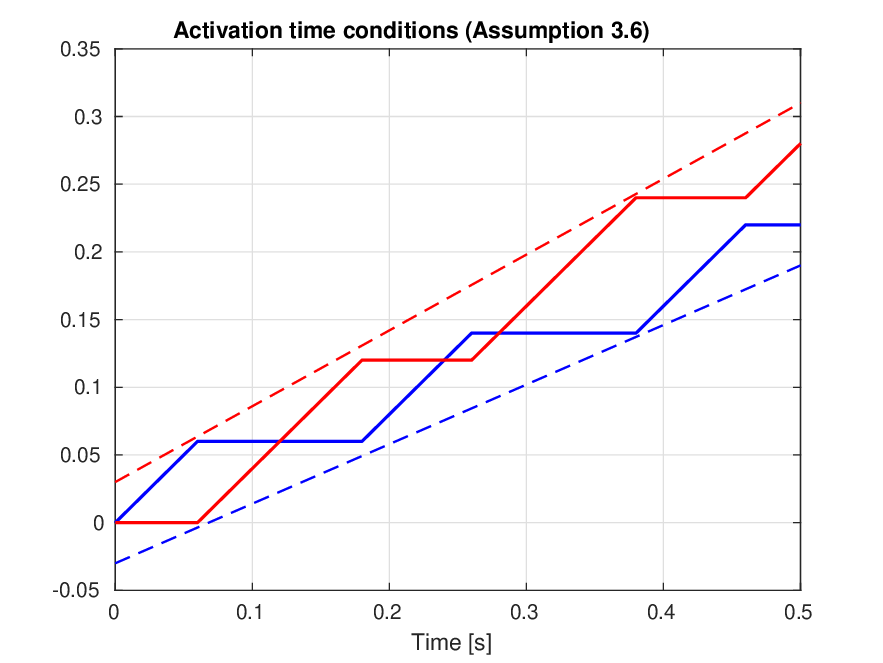}
\vspace{-0.3cm}
\caption{Activation times (solid line) and activation time bounds (dashed line) of Assumption~\ref{ass_times} for the current simulation. 
Subsystem $1$ begins active, with $t_a(1,t,t_0)$ and the lower bound \eqref{tiempo1} in blue and $t_a(2,t,t_0)$ and the upper bound \eqref{tiempo2} in red.}
\label{fig:times}
\end{figure}

\begin{figure}[ht!]
\centering
\includegraphics[width=0.5\columnwidth]{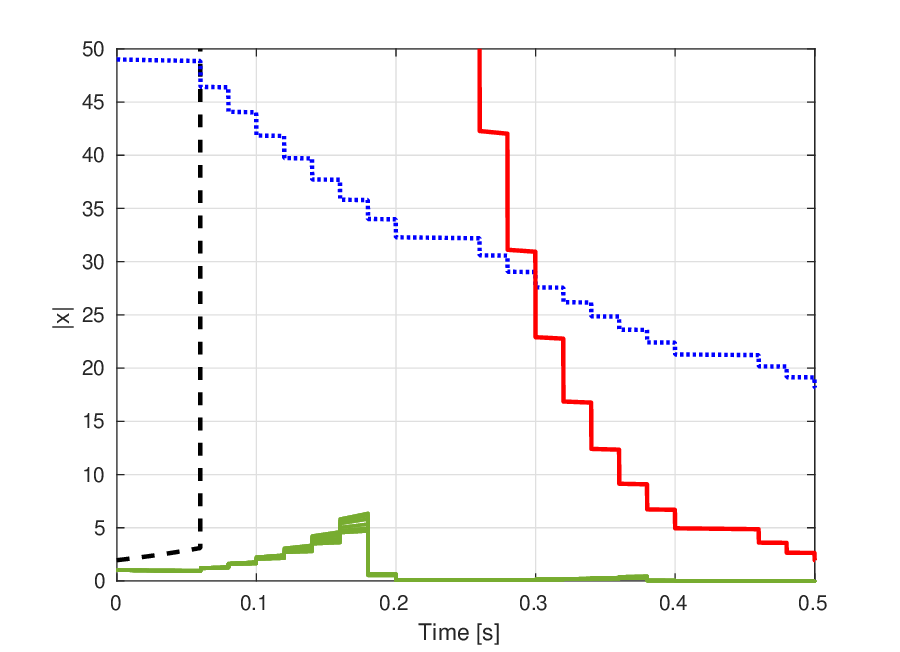}
\vspace{-0.3cm}
\caption{State norm evolution (green) for some systems of $\Sigma_{ \Lambda}$ with $|x_0|=1$. Bounds obtained with $L=1,c_1=0,c_2=1,c_s=0$ (black dashed line), $L=2,c_1=0,c_2=1.22,c_s=0.92$ (blue dotted line) and $L=2,c_1=0.8,c_2=2.3,c_s=0.6$ (red solid line). }
\label{figure2_threebounds}
\end{figure}

\subsection{Application to perturbed systems}
We will show how the presented results can be applied to verify 
input-to-state stability properties for a nonlinear perturbed version of the preceding system. Consider a family $\hat \Sigma$ with dynamics given by
\begin{align}
\label{impulsive_example}
\begin{split}
&\dot x(t)=A_{\nu(t)} x +  \phi_{\nu(t)}(t,x,u), \quad t \notin  \sigma,  \\
&x(t)=J_{\nu(t-),\nu(t)}x(t^-), \phantom{lllll|||||}  \quad t \in  \sigma,
\end{split}
\end{align}
where $x(t)\in \R^2$ and $u$ is a locally bounded and Lebesgue measurable real-valued (i.e. scalar) input function.
$A_1$, $A_2$, $J_{1,2}$ and $J_{2,1}$ are defined 
in \eqref{eq:matrices}
and 
\begin{equation}
\label{eq:nonlinear_part}
\phi_1(t,x,u):=0.01 (1+\sin(t)) |x| + (2|x|+5)u^2, \qquad \phi_2(t,x,u):=0.005(1+\cos(t)) |x| + u.
\end{equation}

The jump graph corresponds to the one shown in Figure~\ref{fig:WDG_example3}, except that for this case $J_{1,1}=J_{2,2}=I$.
We assume that the impulsive behaviour occurs at switching times, thus $\mu=\emptyset$ and $n_\mu(t,t_0) \equiv 0$.
As \eqref{impulsive_example} 
admits a global state weak linearization \citep[Definition II.1]{mancilla2021integral}, we can use the results presented in
 \cite[Theorem III.1]{mancilla2021integral} to establish that it is Strongly integral Input-to-State Stable (S-iISS).
In order to use the mentioned theorem,
we first need to compute the constants $K$ and $\lambda$ that define the S-GUES bound of the linear part of \eqref{impulsive_example}, denoted as $\hat \Sigma^{lin}$.
We assume that the constraints in assumptions~\ref{assumption32},~\ref{ass:jumps} and~\ref{ass_times} are satisfied with the same values presented in the previous example.

Note that neither $J_{1,1}$ nor $J_{2,2}$ is Schur.
\ajv{Following Appendix~\ref{app:boundingLSI}, select} $\I_c=1$, $\I_u=2$, $Q_1=I$ and $P_2=I$ 
and apply Lemma~\ref{property3} to $\hat \Sigma^{lin}$. Then 
 Assumption~\ref{ass031} is satisfied with
$\bar \lambda(1)=-0.7727$, $\bar \lambda(2)=8.8284$, 
$\bar r(1,1)=\bar r(2,2)=1$, $\bar r(1,2)=4.1712$ and $\bar r(2,1)=0.0156$. Given that $\bar r(1,1)=\bar r(2,2)=1$, we have $\lambda_1(1)= \bar \lambda(1)$ and
$\lambda_2(1)= \bar \lambda(2)$. According to the signs of $\lambda_1(1)$ and $\lambda_2(1)$ and 
the definitions 
in \eqref{eq:sets_times}, we have
$\I^u=\I_1=\{2\}$ and $\I^s=\I_2=\{1\}$, 
and Assumption~\ref{ass_times} is satisfied.
Consider $L=\{1,2\}$, from \eqref{Rell} we obtain $R(1)= 4.1712$ and $R(2)=0.0650$. Select $L=2$ and derive $\hat R(2)=4.1712$ from \eqref{hatR} and $T_S=0.1$,  $N_S=c_s \munderbar N_0 -2$ from \eqref{N0TSnus}.
Select $c_s=0.4$. By  Theorem~\ref{cor:SGUES_simp}, using
\eqref{H2-simp} and \eqref{H1-simp} 
we conclude that
the family $\hat \Sigma_\Lambda^{lin}$ is S-GUES  with $K=22.4229$ and $\lambda=0.41$.

In order to apply \cite[Theorem III.1]{mancilla2021integral}, define
 $\tilde N:=0.012 <\frac{\lambda}{Ke^{\lambda}}\approx 0.0121$ and 
\begin{equation*}
    \theta^1(t):=\max\{0.001 (1+\sin(t))-\tilde N,0\}=0, \quad  \theta^2(t):=\max\{0.005 (1+\cos(t))-\tilde N,0\}=0.
\end{equation*}
Given that $\sup_{(\nu,\mu)\in   \Lambda} \lp \int_0^\infty \theta^\nu(s) ds + \sum_{t \in \sigma} \theta^\nu(t)   \rp=0<\infty$, the conditions of the theorem are satisfied and consequently the family
of nonlinear systems of the form \eqref{impulsive_example} %
is S-iISS with respect to input $u$.

\section{Conclusions}
\label{sec:conclusions}
We have established sufficient conditions to ensure strong 
global uniform exponential stability of switched impulsive systems, \hh{where `strong' refers to the fact that the convergence bound decays not only as continuous time elapses but also as the number of jumps increses. These conditions allow} for explicit computation of the initial bound $K$ and the decay rate $\lambda$. These conditions are based on multiple Lyapunov-type functions.
The approach accounts for the existence of unstable dynamics in both the continuous and discrete parts of each switching mode or subsystem.
The given sufficient conditions are based on direct and reverse average dwell-time bounds on the switching signals and impulsive jumps occurring within each mode, and on bounds on the activation times of specific subsystem groups.
The approach allows to adequately consider constraints on admissible jumps between modes and shows that the decay rate estimate can be improved by analyzing longer switching sequences between subsystems. 
For linear switched impulsive systems, we also provide a method for computing the Lyapunov-type functions that satisfy the required assumptions.
The present results extend and complement previous results on specific classes of switched time-varying impulsive nonlinear systems. 
Future work can be aimed at the search for appropriate functions to bound each subsystem's dynamics, the optimization of the coefficients that balance the impact of the flow and jumps of each subsystem to achieve better bounds, and the application of the present results to controller design.

\section*{Appendix}
\begin{appendices}

\section{Bounding functions for linear switched impulsive systems}
\label{app:boundingLSI}
We will provide a constructive method
to derive mode-dependent quadratic functions
that satisfy Assumption \ref{ass031}.
We consider a linear version of \eqref{impulsive_new}, 
i.e. Linear Switched Impulsive (LSI) systems $\Sigma^{lin}$ of the form
\begin{subequations}
\label{impulsive_linear2}
\begin{eqnarray}
&\dot x(t)=A_{\nu(t)} x(t), \phantom{1111} \quad &t \notin  \sigma, \\
&x(t)=J_{{\nu(t^-)},{\nu(t)}} x(t^-),  \quad &t \in  \sigma, \label{am2}
\end{eqnarray}
\end{subequations}
with $A_i, J_{i,j} \in \R^{n\times n}$ for all $i\in \I$ and $(i,j) \in \J \cup \bJ$.
The evolution of $\Sigma^{lin}$ within switching mode $i$ is given by the flow matrix $A_i$ and the single-mode impulse matrix $J_{i,i}$. Taking into account that $A_i$ may be Hurwitz ($H$) or not Hurwitz ($\bar H$), and that $J_{i,i}$ may be Schur ($S$), not Schur ($\bar S$) or identity ($I$), 
we distinguish the five combinations %
$(H,I)$, $(H,S)$, $(H, \bar S)$, $(\bar H,S)$ and $(\bar H, \bar S)$.
 
To compute the functions $V_i$ of Assumption~\ref{ass031},
we partition the index set $\I$ as $\I=\I_c \cup \I_d \cup\I_u$, 
where subsystems of the form  $(H,I), (H , \bar S)$ are in $\I_c$,
subsystems of the form $(\bar H,S)$ are in $\I_d$ and those of the form  $(\bar H, \bar S)$ are in $\I_u$. Subsystems of the form
$(H,S)$ can be chosen to belong to either $\I_c$ or $\I_d$
depending on which choice %
maximizes the decay rate. 
The following lemma shows how to compute quadratic functions $V_i$ %
that satisfy Assumption~\ref{ass031}, for an LSI system of the form \eqref{impulsive_linear2}. 
\jlma{For $P,Q \in \R^{n\times n}$ such that $Q=Q^T$ and $P=P^T>0$ let 
 $\lambda_m(Q,P):= \min \{ \lambda \in \R :\det(Q-\lambda P)=0 \}$,
  $\lambda_M(Q,P):= \max \{ \lambda \in \R :\det(Q-\lambda P)=0 \}$ and let $\lambda_m(P)$ and $\lambda_M(P)$ be the minimum and maximum eigenvalue of $P$, respectively.}
\begin{lemma}
\label{property3}
Select   $Q_i\in \R^{n\times n}$ with $Q_i=Q_i^T>0$,
for all $i\in \I_c \cup \I_d $ and
generate  $P_i=P_i^T>0$ as the unique solution of 
\begin{subequations}
\label{lyap4}
\begin{align}
    &A_i^TP_i+P_iA_i=-Q_i, \quad \forall i \in \I_c, \label{lyap4_c} \\ 
    &J_{i,i}^TP_i J_{i,i}-P_i=-Q_i, \quad \forall i \in \I_d. \label{lyap4_d}
\end{align}
\end{subequations}
Select $P_i=P_i^T>0$ for all $i\in \I_u$
and define 
\begin{equation}
\tilde Q_{i}:=- \lp A_i^T P_i+P_i A_i \rp, \quad \forall i \in \I_d \cup\I_u. \label{tildeQ}
\end{equation}
Define 
the functions $V_i(\xi):=\xi^T P_i \xi$
for all $i\in \I$ and $\xi \in \R^n$.
Then Assumption~\ref{ass031} holds for the LSI system \eqref{impulsive_linear2} with \jlma{$m=2$} and
\begin{align}
\textcolor{black}{\munderbar{K}_i}&:= \lambda_m(P_i), \quad \textcolor{black}{\bar{K}_i}:= \lambda_M(P_i), \quad \forall i\in \I, \label{prop1_x}\\
 \bar \lambda(i)&:=
\begin{cases}
& %
-\lambda_{m}(Q_i,P_i)
\quad \forall i \in \I_c, \\ 
& 
-\lambda_{m}(\tilde Q_i,P_i)
\quad \forall i \in \I_d \cup \I_u, \label{prop2_2} \\
\end{cases}\\
\bar r(i,j)&:=
\begin{cases}
&
\lp 1- \lambda_m(Q_i,P_i) \rp 
\phantom{111} \forall (i,j) \in \bJ : i \in \I_d,  \\
&\lambda_M(J_{i,j}^T P_j J_{i,j},P_i)  \phantom{111}
\text{otherwise },
\end{cases}  
\qquad\text{$\forall(i,j)\in \J\cup \bJ$.}  \label{prop2_3}
\end{align}
\end{lemma}

\begin{proof}
For all %
$x \in \R^n \backslash\{0\}$ %
and $t\notin \sigma$
we have
\begin{equation}
\dot V_i(x)=\nabla V_i(x) A_i x=x^T (A_i^TP_i+P_i A_i) x=-x^T Q_i x =  -\frac{x^T Q_i x}{x^T P_i x} V_i(x) \leq - \min_{x\neq 0} \left \{\frac{x^T Q_i x}{x^T P_i x} \right \} V_i(x)  = -\lambda_{m}(Q_i,P_i) V_i(x),
\label{eq:p_l3.14_1} 
\end{equation}
for all $i \in \I_c$ and
\begin{equation} 
\dot V_i(x) = - x^T \tilde Q_{i} x = -\frac{x^T  \tilde Q_{i} x}{x^T  P_{i} x} V_i(x)   \leq 
- \min_{x\neq 0} \left \{\frac{x^T \tilde Q_i x}{x^T P_i x} \right \} V_i(x)  =
-\lambda_m( \tilde Q_i,P_i) V_i(x), \label{eq:p_l3.14_2}
\end{equation}
for all $i \in \I_d \cup \I_u$, where in \eqref{eq:p_l3.14_1} and \eqref{eq:p_l3.14_2}  we use the property that for $Q_i=Q_i^T \in \R^{n\times n}$ and $P_i=P_i^T>0$,
$P\in \R^{n\times n}$ the minimum generalized eigenvalue is given by $\lambda_m(Q_i,P_i)=\min_{x\neq 0} \left \{\frac{x^T  Q_i x}{x^T P_i x} \right \}$, see e.g. \citep[Section 2.4.3]{kong2017principal}.
Thus, \eqref{ass03-2} holds with $\bar \lambda$ as defined in \eqref{prop2_2}.
For all  $i,j \in \I_d$ such that $i=j$ by \eqref{prop2_3} we have 
$V_i=V_j$ and 
\begin{equation}
V_i(J_{i,j}x)-V_i(x) = - x^T Q_i x = -\frac{x^T Q_i x }{x^T P_i x} V_i(x)
\leq- \lambda_m(Q_i,P_i) V_i(x),
\label{gbe}
\end{equation}
for all $x\in \R^n \backslash \{0\}$.
Combining \eqref{gbe} with \eqref{prop1_x} we obtain
\begin{equation}
 V_j(J_{i,j}x) \leq \lp 1- \lambda_m(Q_i,P_i) \rp V_i(x) = \bar r(i,j) V_i(x).
\end{equation}
Additionally, for some event time $\tau \in \sigma$  
and $i$ and $j$ the active subsystems of $\nu$ before and after $\tau$, respectively,
we have for all $x \in \R^n \backslash \{0\}$ that
\begin{align}
V_j(J_{i,j} x )= \frac{V_j(J_{i,j} x )}{V_i(x)}V_i(x) = \frac{x^TJ_{i,j}^TP_jJ_{i,j}x}{x^TP_i x} V_i(x)\leq \max_{x \neq 0} \left\{ \frac{x^TJ_{i,j}^TP_jJ_{i,j}x}{x^TP_i x}  \right\} V_i(x)= \lambda_M(J_{i,j}^TP_jJ_{i,j},P_i) V_i(x).
\end{align}
Thus, \eqref{ass03-3} holds with $\bar r$ as defined in \eqref{prop2_3}.
\end{proof}
In Lemma~\ref{property3}, we choose $Q_i$ for all $i\in \I_c \cup \I_d$ 
to ensure that $\bar\lambda(i)<0$ for each Hurwitz subsystem $i\in \I_c$ and $r(i,i)<1$ for each Schur subsystem $i\in \I_d$. 
 This selection allows to compute matrices $P_i$ for all $i\in \I_c \cup \I_d$ according to \eqref{lyap4}. Furthermore, we select
$P_i=P_i^T>0$ for all unstable subsystems $i\in \I_u$, ensuring Assumption~\ref{ass031} is satisfied with 
 functions $V_i(\xi)=\xi^T P_i \xi$ for all $i\in \I_c \cup \I_d \cup \I_u$. It is important to note that 
 we may have cases where $\bar \lambda(i)>0$ for some $i \in \I_d \cup \I_u$ and
 $ \bar r(i,j)>1$ for some $(i,j)\in \J$ or $(i,i) \in \bJ$ such that $i\notin \I_d$.

\section{Proof of Lemma~\ref{cor:1}}
\label{appendix:0}
We begin by proving the following claim.
\begin{claim}
\label{claim1}
For all $x_0\in \R^n$ and $t\geq t_0\geq 0$ the solutions with $x(t_0) = x_0$ satisfy 
\begin{equation}
   |x(t)| \leq    \underline{\alpha}_{\nu(t)}^{-1} \lp \exp \lp{ \int_{t_0^+}^t\lambda(\tau)+ r(\tau) d\tau}  \rp \bar \alpha_{\nu(t_0)}(|x_0|)  \rp. \label{eqTH1-x}
\end{equation}
\end{claim}
\begin{proof}[Sketch of proof of Claim~\ref{claim1}]
\let\qed\mer
For the initial time $t_0\ge 0$ and the initial condition $x_0 \in \R^n$ define
the piecewise continuous right\-/continuous function $w : [t_0,\infty) \rightarrow \R_{\geq 0}$ via
$w(t):=V_{\nu(t)}(x(t))$.
From Assumption~\ref{ass031} and the definition of function $\lambda$ in \eqref{func_lambda3} we have
\begin{subequations}
\label{comparison}
\begin{eqnarray}
\dot w(t) \leq  \lambda(t) w(t) \phantom{|||||||||||||||||||} &t \notin  \sigma, \label{eq1} \\
 w(t) \leq  \bar r(\nu(t^-),\nu(t)) w(t^-) &t \in  \sigma. \label{eq2}
\end{eqnarray}
\label{eqcom}
\end{subequations}
For the sake of notation define  $\bar r_k:=\bar r(\nu({\tau_k}^-),\nu({\tau_k}))$.
Define  $\ell := \min \{i \in \N: \tau_i \in \sigma \land \tau_i >t_0 \}$ and, for $t>t_0$, $\bar \ell=\bar \ell(t)=\max \{i \in \N: \tau_i \in \sigma \land \tau_i \le t \}$
so that $\tau_\ell$ represents the first event instant after $t_0$ and $\tau_{\bar \ell}$ the last even instant before or equal to $t$.
Considering the system given by the equality in \eqref{comparison} and
following a similar procedure as in \cite[Theorem 1.4.1]{lakshmikantham1989theory} we get
\begin{equation}
w(t)\leq \lp \prod_{k=\ell}^{\bar \ell} \bar r_k \rp  e^{\int_{t_0^+}^{t} \lambda(s) ds}  w(t_0) = e^{ \sum_{k=\ell}^{\bar \ell} \ln{(\bar r_k)}}  e^{\int_{t_0^+}^{t} \lambda(s) ds}  w(t_0)
 =e^{\int_{t_0^+}^{t} r(s)+\lambda(s) ds} w(t_0).
\label{eq:bound2}
\end{equation}
Thus, 
from \eqref{eq:bound2} and considering \eqref{ass03-1} we have, for all $t\geq t_0$
\begin{equation}
   |x(t)| \leq \underline{\alpha}_{\nu(t)}^{-1} \lp w(t) \rp \leq    \underline{\alpha}_{\nu(t)}^{-1} \lp e^{ \int_{t_0^+}^t\lambda(s)+ r(s) ds} w(t_0) \rp 
   \leq    \underline{\alpha}_{\nu(t)}^{-1} \lp \exp \lp{ \int_{t_0^+}^t\lambda(s)+ r(s) ds} \rp \bar \alpha_{\nu(t_0)}(|x(t_0)|) \rp.
\end{equation}  
\end{proof}

Define $\munderbar{\alpha}^{-1}:= \max_{i\in \I}\{ \munderbar{\alpha}_i^{-1}\}$, 
$\bar\alpha :=  \max_{i\in \I} \bar \alpha_i $ and $\beta \in \KL$ as $\beta(r,s):=\munderbar{\alpha}^{-1}\left(e^C e^{-\lambda_0(s)}\bar \alpha(r)  \right)$. Let $(\nu,\mu)\in \Lambda$.
Thus, from \eqref{condi2c} and \eqref{eqTH1-x}
we have
\begin{align}
|x(t)|&\leq \munderbar{\alpha}_{\nu(t)}^{-1} \lp \exp(C) \exp(- \lambda_0(t-t_0+n(t,t_0)) ) \bar {\alpha}_{\nu(t_0)}(|x_0|)  \rp 
\notag \\
&\leq \munderbar{\alpha}^{-1} \lp \exp(C) \exp(- \lambda_0(t-t_0+n(t,t_0)) ) \bar {\alpha}(|x_0|) \rp = \beta(|x_0|,t-t_0+n(t,t_0)).  \label{ecua}
\end{align}
Now, we will prove that with the functions $\munderbar{\alpha}_i^{-1}, \bar \alpha_i \in \K_\infty$ defined as in Assumption~\ref{ass031}, $\Sigma_{\Lambda}$ is S-GUES.
From the definitions of $\munderbar{\alpha}^{-1}$ and $\bar \alpha$ we have
 $\munderbar \alpha^{-1}(s)=  \lp \frac {s}{\min_{i\in \I} \munderbar K_i}\rp^{1/m}$
and  $\bar \alpha(s)= \lp \max_{i\in \I} \bar K_i \rp s^m$. 
Replacing those functions %
in \eqref{ecua} we obtain
\begin{equation}
|x(t)|\leq  \exp\lp\frac{C}{m}\rp \lp \frac{\max_{i\in \I} \bar K_i}{\min_{i\in \I} \munderbar K_i}\rp^{1/m}  \exp\lp- \frac{\lambda_0}{m}(t-t_0+n(t,t_0)) \rp|x_0| \label{ecua2} 
= K \exp \lp - \lambda (t-t_0+n(t,t_0)) \rp  |x_0|, 
\end{equation}
where 
\begin{equation}
\label{eq:K_lamb}
K:=  \exp\lp\frac{C}{m}\rp \lp\frac{\max_{i\in \I} \bar K_i}{\min_{i\in \I} \munderbar K_i}\rp^{1/m} \quad \text{ and }  \quad \lambda:=\frac{\lambda_0}{m}.
\end{equation}
Thus, $\Sigma_{\Lambda}$ is S-GUES. \qed

\section{Proof of Lemma~\ref{lemma:r_i_SGUAS}}
\label{appendix:2}
i) 
Consider $ \lambda_i(1)<0$.
We proceed by dividing the proof into three cases according to the values of $\ln(\bar r(i,i))$ and $\bar \lambda(i)$.
\begin{enumerate}[1)]
\item  $\ln(\bar r(i,i))<0$ and $\bar \lambda(i)\leq 0$: 
Select $c_i \in(0,1)$, thus 
$r_i(c_i)=(1-c_i)  \ln \lp \bar r(i,i) \rp <0$ and $\lambda_i(c_i)= \bar \lambda(i) +  c_i \frac{\ln \lp \bar r(i,i) \rp}{ T_J^i}<0$. %
\item  $\ln(\bar r(i,i))>0$ and $\bar \lambda(i)<0$: 
With $c_i>1$ we ensure $r_i(c_i)<0$.
Consider $c_i:=1+\epsilon_1$ with $\epsilon_1>0$. We require that 
\begin{equation}
\lambda_i(c_i)=\bar \lambda(i) +  c_i \frac{\ln \lp \bar r(i,i) \rp}{ T_J^i} 
=\lambda_i(1)+ \epsilon_1\frac{\ln \lp \bar r(i,i) \rp}{ T_J^i}<0. 
\label{eq:87}
\end{equation}
From \eqref{eq:87}, we have $0<\epsilon_1 < - \frac{\lambda_i(1) T_J^i}{\ln(\bar r(i,i))}$, this implies that $c_i= 1+\epsilon_1<1 - \frac{\lambda_i(1) T_J^i}{\ln(\bar r(i,i))}=- \frac{\bar \lambda(i) T_J^i}{\ln(\bar r(i,i))}$.
Thus, for all $c_i\in\left (1,- \frac{\bar \lambda(i) T_J^i}{\ln(\bar r(i,i))}\right)$ we have $r_i(c_i)<0$ and $\lambda_i(c_i)<0$.

\item  $\ln(\bar r(i,i))<0$ and $\bar \lambda(i)>0$: With 
$c_i<1$ we ensure
$r_i(c_i)<0$.
Consider $c_i \in (0,1)$. 
From \eqref{eq:87} we have $c_i> 1- \lambda_i(1) \frac{T_J^i}{\ln(\bar r(i,i))}=- \frac{\bar \lambda(i) T_J^i}{\ln(\bar r(i,i))}$. Then for all $c_i\in\left (- \frac{\bar \lambda(i) T_J^i}{\ln(\bar r(i,i))},1\right)$ 
 we have $r_i(c_i)<0$ and $\lambda_i(c_i)<0$.
\end{enumerate}

ii) Consider $\lambda_i(1) \geq 0$. %
We proceed analogously by dividing the problem into three cases.
\begin{enumerate}[1)]
\item  $\ln(\bar r(i,i))<0$ and $\bar \lambda(i)>0$: 
By selecting $c_i \in(0,1)$  we have $r_i(c_i)<0$.
Given that $\hat \lambda(i)>0$ we have $\bar  \lambda(i)\geq -\frac{\ln(\bar r(i,i))}{T^i_J}>-c_i\frac{\ln(\bar r(i,i))}{T^i_J}>0$ for all $c_i\in(0,1)$ and thus $\lambda_i(c_i)=\bar \lambda(i)+c_i\frac{\ln(\bar r(i,i))}{T^i_J}>0$.

\item   $\ln(\bar r(i,i))>0$ and $\bar \lambda(i) \geq 0$: 
 With $c_i \in (1, \infty)$ we ensure $r_i(c_i)<0$. %
Thus, $\lambda_i(c_i)=\bar \lambda(i) +  c_i \frac{\ln \lp \bar r(i,i) \rp}{ T_J^i}>0$.

\item   $\ln(\bar r(i,i))>0$ and $\bar \lambda(i)<0$: 
 Selecting  $c_i \in (1, \infty)$ we have $r_i(c_i)<0$.
Thus, $\lambda_i(c_i)> \lambda_i(1)>0$.  \qed
\end{enumerate}

\section*{Acknowledgments}
Work partially supported by Grant PICT-2021-I-A-00730, Agencia I+D+i, Argentina.

\end{appendices}

\bibliographystyle{elsarticle-harv}
\bibliography{main}

\begin{thebibliography}{35}
\expandafter\ifx\csname natexlab\endcsname\relax\def\natexlab#1{#1}\fi
\providecommand{\url}[1]{\texttt{#1}}
\providecommand{\href}[2]{#2}
\providecommand{\path}[1]{#1}
\providecommand{\DOIprefix}{doi:}
\providecommand{\ArXivprefix}{arXiv:}
\providecommand{\URLprefix}{URL: }
\providecommand{\Pubmedprefix}{pmid:}
\providecommand{\doi}[1]{\href{http://dx.doi.org/#1}{\path{#1}}}
\providecommand{\Pubmed}[1]{\href{pmid:#1}{\path{#1}}}
\providecommand{\bibinfo}[2]{#2}
\ifx\xfnm\relax \def\xfnm[#1]{\unskip,\space#1}\fi
\bibitem[{Briat(2017)}]{briat2017dwell}
\bibinfo{author}{Briat, C.}, \bibinfo{year}{2017}.
\newblock \bibinfo{title}{Dwell-time stability and stabilization conditions for
  linear positive impulsive and switched systems}.
\newblock \bibinfo{journal}{Nonlinear Analysis: Hybrid Systems}
  \bibinfo{volume}{24}, \bibinfo{pages}{198--226}.
\bibitem[{Cai and Teel(2009)}]{cai2009}
\bibinfo{author}{Cai, C.}, \bibinfo{author}{Teel, A.R.}, \bibinfo{year}{2009}.
\newblock \bibinfo{title}{Characterizations of input-to-state stability for
  hybrid systems}.
\newblock \bibinfo{journal}{Systems \& Control Letters} \bibinfo{volume}{58},
  \bibinfo{pages}{47 -- 53}.
\bibitem[{Dayawansa and Martin(1999)}]{dayawansa}
\bibinfo{author}{Dayawansa, W.P.}, \bibinfo{author}{Martin, C.F.},
  \bibinfo{year}{1999}.
\newblock \bibinfo{title}{A converse {L}yapunov theorem for a class of
  dynamical systems which undergo switching}.
\newblock \bibinfo{journal}{IEEE Transactions on Automatic Control}
  \bibinfo{volume}{44}, \bibinfo{pages}{751--760}.
\bibitem[{Feketa et~al.(2021)Feketa, Klinshov and
  L{\"u}cken}]{feketa2021survey}
\bibinfo{author}{Feketa, P.}, \bibinfo{author}{Klinshov, V.},
  \bibinfo{author}{L{\"u}cken, L.}, \bibinfo{year}{2021}.
\newblock \bibinfo{title}{A survey on the modeling of hybrid behaviors: How to
  account for impulsive jumps properly}.
\newblock \bibinfo{journal}{Communications in Nonlinear Science and Numerical
  Simulation} \bibinfo{volume}{103}, \bibinfo{pages}{105955}.
\bibitem[{Gao and Wang(2016)}]{gao2016input}
\bibinfo{author}{Gao, L.}, \bibinfo{author}{Wang, D.}, \bibinfo{year}{2016}.
\newblock \bibinfo{title}{Input-to-state stability and integral input-to-state
  stability for impulsive switched systems with time-delay under asynchronous
  switching}.
\newblock \bibinfo{journal}{Nonlinear Analysis: Hybrid Systems}
  \bibinfo{volume}{20}, \bibinfo{pages}{55--71}.
\bibitem[{Gao et~al.(2015)Gao, Wang and Wang}]{gao2015further}
\bibinfo{author}{Gao, L.}, \bibinfo{author}{Wang, D.}, \bibinfo{author}{Wang,
  G.}, \bibinfo{year}{2015}.
\newblock \bibinfo{title}{Further results on exponential stability for
  impulsive switched nonlinear time-delay systems with delayed impulse
  effects}.
\newblock \bibinfo{journal}{Applied Mathematics and Computation}
  \bibinfo{volume}{268}, \bibinfo{pages}{186--200}.
\bibitem[{Goebel et~al.(2009)Goebel, Sanfelice and Teel}]{goesan_csm09}
\bibinfo{author}{Goebel, R.}, \bibinfo{author}{Sanfelice, R.},
  \bibinfo{author}{Teel, A.R.}, \bibinfo{year}{2009}.
\newblock \bibinfo{title}{Hybrid dynamical systems}.
\newblock \bibinfo{journal}{Control Systems Magazine} \bibinfo{volume}{29},
  \bibinfo{pages}{28--93}.
\bibitem[{Guan et~al.(2005)Guan, Hill and Shen}]{guan2005hybrid}
\bibinfo{author}{Guan, Z.H.}, \bibinfo{author}{Hill, D.J.},
  \bibinfo{author}{Shen, X.}, \bibinfo{year}{2005}.
\newblock \bibinfo{title}{On hybrid impulsive and switching systems and
  application to nonlinear control}.
\newblock \bibinfo{journal}{IEEE Transactions on Automatic Control}
  \bibinfo{volume}{50}, \bibinfo{pages}{1058--1062}.
\bibitem[{Haimovich and Mancilla-Aguilar(2020)}]{haimovich2020nonrobustness}
\bibinfo{author}{Haimovich, H.}, \bibinfo{author}{Mancilla-Aguilar, J.L.},
  \bibinfo{year}{2020}.
\newblock \bibinfo{title}{Nonrobustness of asymptotic stability of impulsive
  systems with inputs}.
\newblock \bibinfo{journal}{Automatica} \bibinfo{volume}{122},
  \bibinfo{pages}{109238}.
\bibitem[{He and Xu(2020)}]{he2020exponential}
\bibinfo{author}{He, D.}, \bibinfo{author}{Xu, L.}, \bibinfo{year}{2020}.
\newblock \bibinfo{title}{Exponential stability of impulsive fractional
  switched systems with time delays}.
\newblock \bibinfo{journal}{IEEE Transactions on Circuits and Systems II:
  Express Briefs} \bibinfo{volume}{68}, \bibinfo{pages}{1972--1976}.
\bibitem[{Ju et~al.(2022)Ju, Meng and Sun}]{ju2022exponential}
\bibinfo{author}{Ju, Y.}, \bibinfo{author}{Meng, F.}, \bibinfo{author}{Sun,
  Y.}, \bibinfo{year}{2022}.
\newblock \bibinfo{title}{Exponential stability of switched linear impulsive
  time-varying system and its application}.
\newblock \bibinfo{journal}{Journal of the Franklin Institute}
  \bibinfo{volume}{359}, \bibinfo{pages}{5619--5633}.
\bibitem[{Kong et~al.(2017)Kong, Hu and Duan}]{kong2017principal}
\bibinfo{author}{Kong, X.}, \bibinfo{author}{Hu, C.}, \bibinfo{author}{Duan,
  Z.}, \bibinfo{year}{2017}.
\newblock \bibinfo{title}{Principal component analysis networks and
  algorithms}.
\newblock \bibinfo{publisher}{Springer}.
\bibitem[{Kundu et~al.(2016)Kundu, Chatterjee and
  Liberzon}]{kundu2016generalized}
\bibinfo{author}{Kundu, A.}, \bibinfo{author}{Chatterjee, D.},
  \bibinfo{author}{Liberzon, D.}, \bibinfo{year}{2016}.
\newblock \bibinfo{title}{Generalized switching signals for input-to-state
  stability of switched systems}.
\newblock \bibinfo{journal}{Automatica} \bibinfo{volume}{64},
  \bibinfo{pages}{270--277}.
\bibitem[{Lakshmikantham et~al.(1989)Lakshmikantham, Simeonov
  et~al.}]{lakshmikantham1989theory}
\bibinfo{author}{Lakshmikantham, V.}, \bibinfo{author}{Simeonov, P.S.}, et~al.,
  \bibinfo{year}{1989}.
\newblock \bibinfo{title}{Theory of impulsive differential equations}.
  volume~\bibinfo{volume}{6}.
\newblock \bibinfo{publisher}{World scientific}.
\bibitem[{Li et~al.(2018)Li, Li and Wang}]{li2018input}
\bibinfo{author}{Li, X.}, \bibinfo{author}{Li, P.}, \bibinfo{author}{Wang,
  Q.g.}, \bibinfo{year}{2018}.
\newblock \bibinfo{title}{Input/output-to-state stability of impulsive switched
  systems}.
\newblock \bibinfo{journal}{Systems \& Control Letters} \bibinfo{volume}{116},
  \bibinfo{pages}{1--7}.
\bibitem[{Li et~al.(2005)Li, Soh and Wen}]{li2005switched}
\bibinfo{author}{Li, Z.}, \bibinfo{author}{Soh, Y.}, \bibinfo{author}{Wen, C.},
  \bibinfo{year}{2005}.
\newblock \bibinfo{title}{Switched and impulsive systems: Analysis, design and
  applications}. volume \bibinfo{volume}{313}.
\newblock \bibinfo{publisher}{Springer Science \& Business Media}.
\bibitem[{Liberzon(2003)}]{liberzon_book03}
\bibinfo{author}{Liberzon, D.}, \bibinfo{year}{2003}.
\newblock \bibinfo{title}{Switching in systems and control}.
\newblock \bibinfo{publisher}{Boston, MA: Birkhauser}.
\bibitem[{Lin and Antsaklis(2009)}]{lin2009stability}
\bibinfo{author}{Lin, H.}, \bibinfo{author}{Antsaklis, P.J.},
  \bibinfo{year}{2009}.
\newblock \bibinfo{title}{Stability and stabilizability of switched linear
  systems: a survey of recent results}.
\newblock \bibinfo{journal}{IEEE Transactions on Automatic control}
  \bibinfo{volume}{54}, \bibinfo{pages}{308--322}.
\bibitem[{Liu et~al.(2011)Liu, Liu and Xie}]{Liu2011899}
\bibinfo{author}{Liu, J.}, \bibinfo{author}{Liu, X.}, \bibinfo{author}{Xie,
  W.C.}, \bibinfo{year}{2011}.
\newblock \bibinfo{title}{Input-to-state stability of impulsive and switching
  hybrid systems with time-delay}.
\newblock \bibinfo{journal}{Automatica} \bibinfo{volume}{47},
  \bibinfo{pages}{899 – 908}.
\bibitem[{Liu et~al.(2012)Liu, Liu and Xie}]{liu2012class}
\bibinfo{author}{Liu, J.}, \bibinfo{author}{Liu, X.}, \bibinfo{author}{Xie,
  W.C.}, \bibinfo{year}{2012}.
\newblock \bibinfo{title}{Class-{KL} estimates and input-to-state stability
  analysis of impulsive switched systems}.
\newblock \bibinfo{journal}{Systems \& Control Letters} \bibinfo{volume}{61},
  \bibinfo{pages}{738--746}.
\bibitem[{Liu et~al.(2022)Liu, Russo, Liberzon and Cavallo}]{9599445}
\bibinfo{author}{Liu, S.}, \bibinfo{author}{Russo, A.},
  \bibinfo{author}{Liberzon, D.}, \bibinfo{author}{Cavallo, A.},
  \bibinfo{year}{2022}.
\newblock \bibinfo{title}{Integral-input-to-state stability of switched
  nonlinear systems under slow switching}.
\newblock \bibinfo{journal}{IEEE Transactions on Automatic Control}
  \bibinfo{volume}{67}, \bibinfo{pages}{5841--5855}.
\bibitem[{Liu et~al.(2019)Liu, Zhang, Lu and Hou}]{liu2019exponential}
\bibinfo{author}{Liu, Z.}, \bibinfo{author}{Zhang, X.}, \bibinfo{author}{Lu,
  X.}, \bibinfo{author}{Hou, T.}, \bibinfo{year}{2019}.
\newblock \bibinfo{title}{Exponential stability of impulsive positive switched
  systems with discrete and distributed time-varying delays}.
\newblock \bibinfo{journal}{International Journal of Robust and Nonlinear
  Control} \bibinfo{volume}{29}, \bibinfo{pages}{3125--3138}.
\bibitem[{Mancilla-Aguilar and Garc\'ia(2000)}]{MANCILLAAGUILAR200067}
\bibinfo{author}{Mancilla-Aguilar, J.}, \bibinfo{author}{Garc\'ia, R.},
  \bibinfo{year}{2000}.
\newblock \bibinfo{title}{A converse {L}yapunov theorem for nonlinear switched
  systems}.
\newblock \bibinfo{journal}{Systems \& Control Letters} \bibinfo{volume}{41},
  \bibinfo{pages}{67--71}.
\bibitem[{Mancilla-Aguilar and Haimovich(2020a)}]{manhai_lcss20}
\bibinfo{author}{Mancilla-Aguilar, J.L.}, \bibinfo{author}{Haimovich, H.},
  \bibinfo{year}{2020}a.
\newblock \bibinfo{title}{Converging-input convergent-state and related
  properties of time-varying impulsive systems}.
\newblock \bibinfo{journal}{IEEE Control Systems Letters} \bibinfo{volume}{4},
  \bibinfo{pages}{680--685}.
\bibitem[{Mancilla-Aguilar and Haimovich(2020b)}]{mancilla2020uniform}
\bibinfo{author}{Mancilla-Aguilar, J.L.}, \bibinfo{author}{Haimovich, H.},
  \bibinfo{year}{2020}b.
\newblock \bibinfo{title}{Uniform input-to-state stability for switched and
  time-varying impulsive systems}.
\newblock \bibinfo{journal}{IEEE Transactions on Automatic Control}
  \bibinfo{volume}{65}, \bibinfo{pages}{5028--5042}.
\bibitem[{Mancilla-Aguilar and Haimovich(2022)}]{mancilla2021integral}
\bibinfo{author}{Mancilla-Aguilar, J.L.}, \bibinfo{author}{Haimovich, H.},
  \bibinfo{year}{2022}.
\newblock \bibinfo{title}{(integral) {ISS} of switched and time-varying
  impulsive systems based on global state weak linearization}.
\newblock \bibinfo{journal}{IEEE Transactions on Automatic Control}
  \bibinfo{volume}{67}, \bibinfo{pages}{6918--6925}.
\bibitem[{Molchanov and Pyatnitskiy(1989)}]{MOLCHANOV198959}
\bibinfo{author}{Molchanov, A.}, \bibinfo{author}{Pyatnitskiy, Y.},
  \bibinfo{year}{1989}.
\newblock \bibinfo{title}{Criteria of asymptotic stability of differential and
  difference inclusions encountered in control theory}.
\newblock \bibinfo{journal}{Systems\& Control Letters} \bibinfo{volume}{13},
  \bibinfo{pages}{59--64}.
\bibitem[{Wang et~al.(2013)Wang, Shi, Zuo, Chen and Shao}]{wang2013finite}
\bibinfo{author}{Wang, Y.}, \bibinfo{author}{Shi, X.}, \bibinfo{author}{Zuo,
  Z.}, \bibinfo{author}{Chen, M.Z.}, \bibinfo{author}{Shao, Y.},
  \bibinfo{year}{2013}.
\newblock \bibinfo{title}{On finite-time stability for nonlinear impulsive
  switched systems}.
\newblock \bibinfo{journal}{Nonlinear Analysis: Real World Applications}
  \bibinfo{volume}{14}, \bibinfo{pages}{807--814}.
\bibitem[{Wei et~al.(2019)Wei, Zhi, Liu and Mu}]{WEI20193102}
\bibinfo{author}{Wei, J.}, \bibinfo{author}{Zhi, H.}, \bibinfo{author}{Liu,
  K.}, \bibinfo{author}{Mu, X.}, \bibinfo{year}{2019}.
\newblock \bibinfo{title}{Stability of mode-dependent linear switched singular
  systems with stable and unstable subsystems}.
\newblock \bibinfo{journal}{Journal of the Franklin Institute}
  \bibinfo{volume}{356}, \bibinfo{pages}{3102--3114}.
\bibitem[{Xie et~al.(2013)Xie, Zhang, Zhang and Wang}]{xie2013exponential}
\bibinfo{author}{Xie, D.}, \bibinfo{author}{Zhang, H.}, \bibinfo{author}{Zhang,
  H.}, \bibinfo{author}{Wang, B.}, \bibinfo{year}{2013}.
\newblock \bibinfo{title}{Exponential stability of switched systems with
  unstable subsystems: a mode-dependent average dwell time approach}.
\newblock \bibinfo{journal}{Circuits, Systems, and Signal Processing}
  \bibinfo{volume}{32}, \bibinfo{pages}{3093--3105}.
\bibitem[{Yang et~al.(2018)Yang, Zhao, Bu and Qian}]{yang2018stabilization}
\bibinfo{author}{Yang, J.}, \bibinfo{author}{Zhao, X.}, \bibinfo{author}{Bu,
  X.}, \bibinfo{author}{Qian, W.}, \bibinfo{year}{2018}.
\newblock \bibinfo{title}{Stabilization of switched linear systems via
  admissible edge-dependent switching signals}.
\newblock \bibinfo{journal}{Nonlinear Analysis: Hybrid Systems}
  \bibinfo{volume}{29}, \bibinfo{pages}{100--109}.
\bibitem[{Zhan et~al.(2021)Zhan, Ma, Li and Pedrycz}]{zhan2021exponential}
\bibinfo{author}{Zhan, T.}, \bibinfo{author}{Ma, S.}, \bibinfo{author}{Li, W.},
  \bibinfo{author}{Pedrycz, W.}, \bibinfo{year}{2021}.
\newblock \bibinfo{title}{Exponential stability of fractional-order switched
  systems with mode-dependent impulses and its application}.
\newblock \bibinfo{journal}{IEEE transactions on cybernetics}
  \bibinfo{volume}{52}, \bibinfo{pages}{11516--11525}.
\bibitem[{Zhang et~al.(2019)Zhang, Li, Xu and Li}]{zhang2019stability}
\bibinfo{author}{Zhang, T.}, \bibinfo{author}{Li, J.}, \bibinfo{author}{Xu,
  W.}, \bibinfo{author}{Li, X.}, \bibinfo{year}{2019}.
\newblock \bibinfo{title}{Stability and {$L_2$}-gain analysis for impulsive
  switched systems}.
\newblock \bibinfo{journal}{Communications in Nonlinear Science and Numerical
  Simulation} \bibinfo{volume}{78}, \bibinfo{pages}{104854}.
\bibitem[{Zhao et~al.(2014)Zhao, Yin, Li and Niu}]{zhao2014switching}
\bibinfo{author}{Zhao, X.}, \bibinfo{author}{Yin, S.}, \bibinfo{author}{Li,
  H.}, \bibinfo{author}{Niu, B.}, \bibinfo{year}{2014}.
\newblock \bibinfo{title}{Switching stabilization for a class of slowly
  switched systems}.
\newblock \bibinfo{journal}{IEEE Transactions on Automatic Control}
  \bibinfo{volume}{60}, \bibinfo{pages}{221--226}.
\bibitem[{Zhao et~al.(2012)Zhao, Zhang, Shi and Liu}]{6097035}
\bibinfo{author}{Zhao, X.}, \bibinfo{author}{Zhang, L.}, \bibinfo{author}{Shi,
  P.}, \bibinfo{author}{Liu, M.}, \bibinfo{year}{2012}.
\newblock \bibinfo{title}{Stability and stabilization of switched linear
  systems with mode-dependent average dwell time}.
\newblock \bibinfo{journal}{IEEE Transactions on Automatic Control}
  \bibinfo{volume}{57}, \bibinfo{pages}{1809--1815}.

\end{thebibliography}
\end{document}